\documentclass{svjour3}

\usepackage{amssymb}
\usepackage{epstopdf,comment}
\usepackage{  mathrsfs, enumerate}
\usepackage{amsmath}

\newtheorem{assumption}{Assumption}

\newcommand{\req}[1]{Eq.\,(\ref{#1})}

\begin{document} 

\title{  Entropy Anomaly in Langevin-Kramers Dynamics with a
Temperature Gradient, Matrix Drag, and Magnetic Field }
\titlerunning{Entropy Anomaly in Langevin-Kramers Dynamics}

\author{Jeremiah Birrell}
\institute{J. Birrell \at
Department of Mathematics and Statistics\\
University of Massachusetts Amherst\\
Amherst, MA 01003, USA\\
\email{birrell@math.umass.edu}}

\maketitle

\begin{abstract}
We investigate  entropy production in the small-mass (or overdamped) limit of  Langevin-Kramers dynamics.  The results generalize previous works to provide a rigorous derivation that covers systems with magnetic field as well as anisotropic (i.e. matrix-valued) drag and diffusion coefficients  that satisfy a fluctuation-dissipation relation with state-dependent temperature. In particular, we derive an explicit formula for the anomalous entropy production which can be estimated from simulated paths of the overdamped system.

As a part of this work, we develop a  theory for homogenizing a class of integral processes involving the position and scaled-velocity variables. This allows us to rigorously identify the limit of the entropy produced in the environment, including a bound on the convergence rate. 

\end{abstract}

\section{Introduction}

Langevin-Kramers equations model  the motion of a noisy, damped, diffusing particle of non-zero mass, $m$.  In the  simplest case, the stochastic differential equation (SDE) has the form
\begin{align}\label{model_sys}
dq_t=v_t dt,\hspace{2mm} m dv_t=-\gamma v_t dt+\sigma dW_t,
\end{align}
where $\gamma$ and $\sigma$ are the dissipation (i.e. drag) and diffusion coefficients respectively and $W_t$ is a Wiener process.   Smoluchowski \cite{smoluchowski1916drei} and Kramers \cite{KRAMERS1940284} pioneered the study of such diffusive systems in the small-mass (or overdamped) limit; see  \cite{Nelson1967} for  more on the early literature and \cite{doi:10.1137/S1540345903421076,Chevalier2008,bailleul2010stochastic,pinsky1976isotropic,pinsky1981homogenization,Jorgensen1978,dowell1980differentiable,XueMei2014,angst2015kinetic,bismut2005hypoelliptic,bismut2015}
 for further mathematical results in this direction.

The $m\to 0$ limit of \req{model_sys} is a common and useful approximation for simulating particle paths in many realistic systems.  In the simplest case, setting $m=0$ in \req{model_sys} gives the correct overdamped SDE:
\begin{align}
dq_t=\gamma^{-1}\sigma dW_t.
\end{align}
This naive derivation is known to fail when $\gamma$ is state-dependent, but a slightly more complicated SDE, incorporating an anomalous drift term, does govern the overdamped particle trajectories; see Theorem \ref{theorem:conv} below for a summary of a known convergence result. The study of the singular nature of the Langevin-Kramers system in the small-mass limit (i.e. the appearance of the anomalous drift) has a long history. See, for example,  \cite{PhysRevA.25.1130,Sancho1982,volpe2010influence,Hottovy2014,herzog2015small,particle_manifold_paper}.

 The overdamped SDE from  Theorem \ref{theorem:conv}  correctly describes the limiting behavior of the position variables in a temperature gradient and over bounded time intervals. However, naively applying it outside this domain can lead to  erroneous results.  For example, it is known that in systems with time-dependent diffusion and damping, the overdamped SDE can fail to capture the long-time behavior of the underdamped system; see \cite{BodrovaSciReport} for a detailed analysis of a family of anomalous diffusion processes that exhibit this behavior. Naive approaches to the overdamped limit can also fail for systems  coupled to multiple reservoirs \cite{PhysRevE.94.062148}.

There is also a more subtle way that use of the overdamped SDE can lead to errors; it may provide a good approximation to the statistics of particle paths, yet still fail to capture the statistics of other important observables in stochastic thermodynamics. Mathematically, this occurs when the operations of taking the small-mass limit and computing the observable do not commute. Such an occurrence is typically called an anomaly in the physics literature; see, for example, \cite{fujikawa2004path}. Specifically, this paper will focus on  entropy production, a quantity which has attracted a great deal of study, especially in systems with a (time-dependent) temperature gradient \cite{PhysRevLett.109.260603,PhysRevE.87.050102,doi:10.1063/1.4833136,0295-5075-103-1-10010,PhysRevE.89.022127,10.1038/srep12266,PhysRevE.92.042108,PhysRevE.92.062110,1742-5468-2015-1-P01005,PhysRevE.93.012132,PhysRevE.93.042121,PhysRevE.93.042116,PhysRevE.93.052149,PhysRevLett.117.070601,PhysRevE.97.042112}.  For  general background information on stochastic thermodynamics, see \cite{doi:10.1063/1.533195,1742-5468-2007-09-L09002,Chetrite2008,0034-4885-75-12-126001,gawedzki2013fluctuation}.

 It is known that in systems with a temperature gradient, the  overdamped  limit of Langevin-Kramers dynamics exhibits an entropy anomaly; the entropy production associated with the limiting overdamped SDE has a deficit when compared to the small-mass limit of the entropy produced by the underdamped SDE.    Derived formally in  \cite{PhysRevLett.109.260603}, the entropy anomaly can be traced to a time-reversal symmetry breaking that occurs when transitioning between the under and overdamped systems. This was generalized in  \cite{PhysRevE.93.012132} to allow for Brownian rotations and spatially dependent, matrix-valued drag, as well as external forces and torques. Systems in a mean flow and with  Brownian rotations were studied in \cite{10.1038/srep12266}, but in a different regime than the $m\to 0$ limit considered here.  A similar effect has also been studied in continuous-time Markov chains \cite{PhysRevE.93.052149}.

 In this paper, we put one aspect of the prior formal derivations  of the entropy anomaly on a rigorous footing. Specifically, we  prove a convergence result as $m\to 0$, with a convergence rate bound, for the entropy produced in the environment in Langevin-Kramers dynamics.  The result covers systems in a temperature gradient, allowing for magnetic field and state-dependent,  matrix-valued drag. Using this, we derive an explicit formula for the anomalous contribution to the entropy production in two different classes of models. The formulas are expectations of functionals of the paths of the overdamped system, and so are straightforward to estimate from numerical simulations.

\subsection{Background and Previous Results}\label{sec:prev}

The Hamiltonian of a  particle of mass $m$ in an electromagnetic field with $C^3$ vector potential $\psi(t,q)$ and $C^2$ potential $V(t,q)$ (setting charge $e=1$) is
\begin{align}\label{hamiltonian_family}
H(t,x)=\frac{1}{2 m}\|p-\psi(t,q)\|^2+V(t,q)
\end{align}
where $x\equiv(q,p)\in\mathbb{R}^n\times\mathbb{R}^n$. Allowing for an additional continuous forcing term, $\tilde F$, and coupling to noise and linear drag via the $C^2$ matrix-valued functions $\sigma$ and $\gamma$ respectively, Hamilton's equations for this system are given by the  SDE
\begin{align}
dq_t^m=&\frac{1}{ m}(p_t^m-\psi(t,q_t^m))dt,\label{q_eq}\\
d(p_t^m)_i=&\left(-\frac{1}{ m}\gamma_{ij}(t,q_t^m)\delta^{jk}((p_t^m)_k-\psi_k(t,q_t^m))+\tilde F_i(t,x_t^m)-\partial_{q^i} V(t,q_t^m)\right.\label{p_eq}\\
&\left.+\frac{1}{ m}\partial_{q^i}\psi_k(t,q_t^m)\delta^{jk}((p_t^m)_j-\psi_j(t,q_t^m))\right)dt+\sigma_{i\rho}(t,q_t^m)dW^\rho_t.\notag
\end{align}
Note that the choice of  stochastic integral (It{\^o}, Stratonovich, etc.) is not important in \req{p_eq}; $q_t^m$ is a $C^1$ process and $\sigma$ is $C^2$, hence  all choices yield the same equation.  

It is often convenient to define $u_t^m\equiv p_t^m-\psi(t,q_t^m)$ and write the SDE in the equivalent form
\begin{align}
dq_t^m=&\frac{1}{m }u_t^m dt,\label{q_eq2}\\
d(u^m_t)_i=&\left(-\frac{1}{m}\tilde \gamma_{ik}(t,q_t^m) (u_t^m)^k+  F_i(t,x^m_t)\right)dt+\sigma_{i\rho}(t,q_t^m)dW^\rho_t,\label{u_eq}
\end{align}
where $\tilde \gamma$ now includes the magnetic field,
\begin{align}\label{tilde_gamma_def}
\tilde \gamma_{ik}(t,q)\equiv\gamma_{ik}(t,q)-H_{ik}(t,q)\equiv \gamma_{ik}(t,q) -(\partial_{q^i}\psi_k(t,q)-\partial_{q^k}\psi_i(t,q)),
\end{align}
and
\begin{align}
F(t,x)\equiv-\partial_t\psi(t,q)-\nabla_q V(t,q)+\tilde F(t,x).
\end{align}
Here and in the following we employ the summation convention for repeated indices. 

In this paper we will assume the fluctuation-dissipation relation holds pointwise for a time and state-dependent effective temperature.
\begin{assumption}\label{assump:fluc_dis}
Define
\begin{align}
\Sigma_{ij}(t,q)=\sum_\rho \sigma_{i\rho}(t,q)\sigma_{j\rho}(t,q).
\end{align}
We assume $\gamma$ and $\sigma$ are $C^2$ and
\begin{align}\label{fluc_dis}
 \Sigma(t,q)=2\beta^{-1}(t,q)\gamma(t,q),
\end{align}
where $\beta$ is a $C^2$ function that is bounded above and below by positive constants.   Physically, $\beta$ is related to the time and position-dependent effective temperature by $\beta^{-1}=k_BT$, where $k_B$ is the Boltzmann constant.
\end{assumption}
The above assumption is a generalization of the classical Einstein relation, in which case $\beta$ is a constant, to a  class of non-equilibrium settings. Physically, a temperature gradient can be achieved through coupling to more than one heat reservoir, as in \cite{PhysRevLett.109.260603}. One can also view \req{fluc_dis} as expressing a timescale separation between fast dynamics  that maintain a local approximate equilibrium, characterized by a local effective temperature, and a slow relaxation to global equilibrium. This is a commonly useful class of models, but it is still a restricted type of non-equilibrium and such a description can fail due to lack of timescale separation  \cite{MARCONI2008111,10.1371/journal.pone.0093720}. See \cite{1751-8121-44-48-483001} for further examples and references related to effective temperature and fluctuation-dissipation relations.

Our derivation of the entropy anomaly relies on several previously derived estimates on the small-mass limit of \req{q_eq}-\req{p_eq}. In \cite{BirrellHomogenization} it was shown that, for a large class of systems generalizing  \req{q_eq}-\req{p_eq}, there exists unique global in time solutions $(q_t^m,u_t^m)$ that converge to $(q_t,0)$ as $m\rightarrow 0$, where here $q_t$ is the solution to a certain limiting SDE.  We summarize the precise mode of convergence in Theorem \ref{theorem:conv} below, which we take as the  starting point for this work. See also \cite{Hottovy2014} and further references therein for earlier, related results.    Appendix \ref{app:assump} contains a  list of assumptions that guarantee that the following theorem holds.  
\begin{theorem}\label{theorem:conv}
 For any  $T>0$, $p>0$  we have
\begin{align}\label{results_summary1}
\sup_{t\in[0,T]} E\left[\|u_t^m\|^p\right]^{1/p}=O(m^{1/2}),\hspace{2mm} \sup_{t\in[0,T]}E\left[\|q_t^m-q_t\|^p\right]^{1/p}=O(m^{1/2}) 
\end{align}
as $m\rightarrow 0$, where $q_t$ is the solution the It\^o SDE
\begin{align}\label{q_SDE}
dq_t=&\tilde\gamma^{-1}(t,q_t)F(t,q_t,\psi(t,q_t))dt+S(t,q_t)dt+\tilde\gamma^{-1}(t,q_t)\sigma(t,q_t)dW_t.
\end{align}
$S(t,q)$, called the  {\em noise-induced drift}, is an anomalous drift term that arises in the limit. It is given by
\begin{align}\label{Ito_S}
S^i(t,q)\equiv \beta^{-1}(t,q)\partial_{q^j}(\tilde\gamma^{-1})^{ij}(t,q).
\end{align}
  $q_t$ also satisfies
\begin{align}\label{q_Lp_bound}
E\left[\sup_{t\in[0,T]}\|q_t\|^p\right]^{1/p}<\infty
\end{align}
for all $T>0$, $p>0$.

The components of $\tilde\gamma^{-1}$ are defined such that 
\begin{align}\label{tilde_gamma_inv_def}
(\tilde\gamma^{-1})^{ij}\tilde\gamma_{jk}=\delta^i_k,
\end{align}
and for any $v_i$ we define the contraction $(\tilde\gamma^{-1}v)^i=(\tilde\gamma^{-1})^{ij}v_j$.
\end{theorem}

It should be emphasized that the paths of the overdamped system, $q_t$, are uniquely determined by  the underdamped system, \req{q_eq}-\req{p_eq}, without requiring a choice of stochastic integration convention as additional data. As noted above, the underdamped system is independent of the choice of convention, due to smoothness of $q_t^m$.  Taking the $m\to 0$ limit then uniquely identifies the paths, $q_t$, of the overdamped system.  This point, and the insight it gives into the  physics of systems in a temperature gradient, was noted in \cite{MATSUO2000188}, where it was used to compare  heat generation in  underdamped and overdamped SDEs.

   Of course, one can write the limiting SDE for the overdamped paths, $q_t$, using whichever integration convention is desired, resulting in different formulas for the noise-induced drift. Most commonly, the It\^o form (see \req{q_SDE}-\req{Ito_S}), or Stratonovich form (see \req{q_strat_SDE}-\req{tilde_S_def} below) are used, though anti-It\^o or anything between can be used if desired.  One convention may appear most natural in a given setting, due to the vanishing of the corresponding noise-induced drift (for example, the It\^o convention when $\tilde\gamma$ is independent of $q$), but this is not always possible to accomplish when both temperature and drag are spatially dependent. In any case, as long as the chosen integration convention is consistently paired with its noise-induced drift, one  obtains a SDE that can be solved to find the overdamped paths, and all choices will result in the same solution, $q_t$.

For the purposes of computing entropy production, it will be most convenient to work with the Stratonovich convention, and its corresponding noise-induced drift. Assuming that $\sigma$ and $\gamma$ are $C^2$ and $\psi$ is $C^3$, the overdamped SDE has the Stratonovich form
\begin{align}\label{q_strat_SDE}
dq_t=&\tilde\gamma^{-1}(t,q_t)F(t,q_t,\psi(t,q_t))dt+\tilde S(t,q_t)dt+\tilde\gamma^{-1}(t,q_t)\sigma(t,q_t)\circ dW_t,
\end{align}
where the noise-induced drift in the Stratonovich convention is
\begin{align}\label{tilde_S_def}
\tilde S^i(t,q)=&\beta^{-1}(t,q) \partial_{q^j}(\tilde\gamma^{-1})^{il}(t,q) (\tilde\gamma^{-1})^{jk}(t,q)H_{ lk}(t,q)\\
&-\frac{1}{2} \sum_{\xi}(\tilde\gamma^{-1})^{il}(t,q)\partial_{q^k}\sigma_{l\xi}(t,q) (\tilde\gamma^{-1}(t,q)\sigma(t,q))^k_\xi.\notag
\end{align}

 We will also need  the  convergence result from \cite{Birrell2018}, concerning  the joint distribution of $q_t^m$ and $z_t^m$, where 
\begin{align}
z_t^m\equiv u_t^m/\sqrt{m}.
\end{align}
From \req{q_eq2}, we see that $z_t^m$ is a scaled velocity.

  The properties in  Appendix \ref{app:assump}, along with Assumption \ref{assump:fluc_dis}, imply the following.
\begin{theorem}\label{theorem:conv_dist}
 Let   $K,q>0$, $0<\delta<1/2$, and $\tilde h: \mathbb{R}^n\times \mathbb{R}^n\rightarrow\mathbb{C}$ be a  $C^1$ function that satisfies
\begin{align}\label{nabla_h_bound}
\|\nabla \tilde h(q,z)\|\leq K(1+\|(q,z)\|^q).
\end{align}
Define
\begin{align}
H_t=E\left[ \left(\frac{ \beta(t,q_t)}{2\pi}\right)^{n/2} \int  \tilde h(q_t,z) e^{- \beta(t,q_t)\|z\|^2/2}dz\right].
\end{align}
 Then
\begin{align}\label{dis_limit}
&E\left[\tilde h(q_t^m,z_t^m)\right]=H_t+O(m^{\delta})
\end{align}
as $m\rightarrow 0$.

\end{theorem}
Intuitively, this results states that, as $m\to 0$, the $z$-dependence of the joint distribution of the position, $q_t^m$, and scaled velocity, $z_t^m$, is well approximated by a Gibbs distribution at the local temperature, and one can approximate an observable, $\tilde h$, with its local equilibrium (in $z$) average.

\subsection{Summary of Results}
Our primary result is a formula for the small-mass limit of the entropy produced in the environment  for the  Langevin-Kramers system, \req{q_eq}-\req{p_eq}, including the effects of a temperature gradient,  magnetic field, and matrix-valued drag.  The general result is found in Section \ref{sec:underdamped_entropy_limit},  Theorem \ref{S_conv_thm}, including precise assumptions under which we prove convergence.  Corollary \ref{zero_psi_corollary} contains the simplified case where the vector potential, $\psi$, vanishes.  We quote it here for illustrative purposes:

Let $S^{env,m}_{s,t}$ denote the entropy produced in  the environment over a time interval $[s,t]$. For any $0<\delta<1/2$, and assuming $\psi=0$, we have
 \begin{align}\label{S_env_limit_summary}
&E[S^{env,m}_{s,t}]\\
= &E\left[(\beta V)(s,q_s)\right]-E\left[(\beta V)(t,q_t)\right]+E\left[\int_s^t  \partial_r (\beta V)(r,q_r)dr\right]\notag\\
&+\int_s^tE\left[\beta^{-1}(r,q_r)\nabla_q\cdot \left(\gamma^{-1}\left(V\nabla_q \beta +\beta\tilde F\right) \right)(r,q_r)  \right]dr\notag\\
&+\frac{n+2}{2} E\left[\ln(\beta(t,q_t)/\beta(s,q_s))\right]- E\left[\int_s^t \left(\beta^{-1}\partial_r \beta\right)(r,q_r)dr \right]\notag\\
&+\int_s^t  E\left[\left(\left(-\nabla_q V+\tilde F\right)\cdot \gamma^{-1}\left(V\nabla_q \beta +\beta\tilde F\right)\right)(r,q_r)\right]dr\notag\\
&+\int_s^tE\left[\left( \beta^{-3}\nabla_q\beta\cdot \left( \frac{3n+2}{6}\gamma^{-1}- \int_0^\infty Tr[\gamma e^{-2y \gamma}] \gamma^{-1}e^{-y \gamma} dy\right)\nabla_q\beta\right)(r,q_r)  \right]dr\notag\\
&+O(m^\delta)\notag
\end{align}
as $m\to 0$.  In particular, note that the $m\to 0$ limit can be computed  from the solution paths of the overdamped SDE.

In Section \ref{sec:anom_entropy}, the above formula will be compared with the entropy   production in the environment for the overdamped SDE, \req{q_strat_SDE}. A formal calculation will then result in a formula for the total entropy production in each case, and we will find that the results differ i.e. the operations of computing the entropy production and taking the small-mass limit do not commute.  Specifically, in Section \ref{sec:anomaly}  we identify following deficit in the entropy production of the overdamped SDE, as compared to the small-mass limit of the entropy production of the underdamped SDE:
\begin{align}\label{S_anom_summary}
&E\left[S_{s,t}^{anom}\right]\\
=&\int_s^tE\left[\left( \beta^{-3}\nabla_q\beta\cdot \left( \frac{3n+2}{6}\gamma^{-1}- \int_0^\infty Tr[\gamma e^{-2y \gamma}] \gamma^{-1}e^{-y \gamma} dy\right)\nabla_q\beta\right)(r,q_r)  \right]dr.\notag
\end{align}
Again, this formula applies to the $\psi=0$ case; see \req{S_anom} for the general result, as well as Section \ref{sec:unif_B} for the entropy anomaly in another class of models with magnetic field. As previously noted, a version of this anomalous entropy production (without magnetic field) was first derived formally in \cite{PhysRevLett.109.260603} and generalized in \cite{PhysRevE.93.012132}.  Our treatment puts one aspect of these derivations on a rigorous footing; we prove convergence of the entropy produced in the environment as $m\to 0$, including an explicit convergence rate bound. We also extend the result to cover systems in a magnetic field.

The main new technical contribution of this paper is a method for computing the small-mass limit (i.e. homogenization) of certain integral processes of the form $\int_s^t G(r,q_r^m,z^m_r)dr$ and $m^{-1/2}\int_s^t z^m_r\cdot  K(r,q_r^m,z^m_r)dr$ that are multi-linear in $z$, with $K$ being even in $z$.  This is done in Appendix \ref{sec:homog}; see Theorems \ref{homog_thm_1} and \ref{homog_thm_2}.

\section{Background: Time-Inversion and Entropy Production}\label{sec:entropy}

In this section we present a synopsis of the theory of time-inversion and entropy production in stochastic thermodynamics, using the framework   in \cite{Chetrite2008,gawedzki2013fluctuation}.  See also \cite{PhysRevE.97.042112} for a discussion of the relationship between entropy production and the system-bath interaction.
\subsection{Time-Inversion}\label{sec:time_inv}
 Consider a generic SDE in Stratonovich form
\begin{align}\label{original_SDE}
dx_t=b(t,x_t)dt+\tilde \sigma(t,x_t)\circ dW_t
\end{align}
on the time interval $[0,T]$, driven by a Wiener process, $W_t$, and smooth drift and diffusion, $b$ and $\tilde \sigma$.  

A  time-inversion (or time-reversal) operation on spacetime will be given by a map $(t,x)\rightarrow  (t^*,x^*)$ where $x\to x^*$ (which we will also write as $\phi(x)$ and refer to as a time-inversion) is a smooth involution and $t^*=T-t$. For example, if $x=(q,p)$ has position and momentum components $q$ and $p$ respectively, a common choice is $(q,p)^*=(q,-p)$. Note that we do {\em not} assume time-inversion leaves \req{original_SDE} invariant; in general, dissipative terms ensure it is not invariant.

One could define the time-reversed trajectories of the original system \req{original_SDE} by $\tilde x_t=x^*_{t^*}$, however this is  problematic as, for example, it leads to anti-dissipation.  A more physically reasonable method of defining the time-reversed dynamics is to split the drift into two components $b=b_++b_-$ (called the dissipative and conservative parts, respectively \cite{Chetrite2008}) and define the time-reversed  process to be the solution to the SDE
\begin{align}\label{gen_time_inv_rule}
dx^\prime_t=(\phi_* b_+)(t^*,x^\prime_t)dt-(\phi_*b_-)(t^*,x^\prime_t)dt+ (\phi_*\tilde\sigma)(t^*,x^\prime_t)\circ d\tilde W_t,
\end{align}
where $\phi_*$ denotes the pushforward of vector fields by the smooth map $\phi$ (i.e. the operation that takes a vector field and transforms it under the coordinate transformation $x\to \phi(x)$) and $\tilde W_t$ is any other Wiener process. We call the solution $x^\prime_t$ the  backward process while $x_t$ will be called the  forward process.

In \req{gen_time_inv_rule}, the noise term is kept the `same' for both the forward and backward processes, only transformed by $\phi$ to the new coordinate system.  The choice of driving Wiener process doesn't impact the distribution of the solution and so can be chosen based on convenience. We also note that defining the drift, and its splitting,  via the Stratonovich form of the SDE can be motivated by the fact that only the Stratonovich integral has the correct transformation property under change of coordinates \cite{hsu2002stochastic}. For the application to the underdamped system \req{q_eq}-\req{p_eq}, all stochastic integral choices lead to the same drift (see the comment after \req{p_eq}) and so this point is inconsequential there.

Physically, the procedure outlined above is often carried out in the opposite order;  one  has physically motivated forward and time-reversed SDEs for $x_t$ and $x^\prime_t$ and a phase-space involution, $\phi$, and one wants to reverse engineer a splitting of the drift so that the SDEs correspond as in \req{gen_time_inv_rule}.  In any case, the pair of forward and backward equations should capture the physics of what one wants to call time-reversal in a given system.

\subsection{Entropy Production}
The entropy produced by the process \req{original_SDE} relative to the time-reversed process \req{gen_time_inv_rule} is defined via the Radon-Nikodym derivative of the distribution of the backward process with respect to the forward process; see \cite{Chetrite2008} for details.  Intuitively, it quantifies how likely paths are for the backward process, as compared to the forward process.

Specifically, the entropy produced in the environment from time $s$ to time $t$, $S^{env}_{s,t}$ can be computed via the formula 
\begin{align}\label{S_env_def}
S^{env}_{s,t}\equiv &\int_{s}^t 2\hat b_+^j(r,x_r)  (\tilde \Sigma^{-1})_{jk}(r,x_r)\circ dx^k_r\\
&-\int_{s}^t  2\hat b_+^j(r,x_r) (\tilde \Sigma^{-1})_{jk}(r,x_r)b^k_-(r,x_r)  +\nabla\cdot b_-(r,x_r)dr,\notag
\end{align}
where $\tilde\Sigma=\tilde\sigma\tilde\sigma^T$, $\tilde\Sigma^{-1}$ is the pseudoinverse, $\circ$ denotes the Stratonovich integral, and
\begin{align}
\hat b_+^i\equiv b_+^i-\frac{1}{2}\sum_\xi \tilde\sigma^i_\xi\partial_j\tilde\sigma^j_\xi.
\end{align}
Note that this only depends on the splitting of the drift, and not on the form of the time-inversion map.

\section{Time-Inversion in Langevin Dynamics}\label{sec:underdamp_time_inv}

The  definition of entropy production outlined in the prior section is quite flexible and powerful, but can be physically ambiguous; given a forward equation, the framework it does not uniquely single out a time-inversion map, nor does it single out a preferred splitting of the drift.  These two choices should be determined by the physics of what one would like to call the time-reversed process.  In many cases the proper choices appear obvious, but in others (such as Langevin dynamics in a magnetic field) they are less clear.

  The intuition  guiding our choices of inversion and splitting  is that the backward equation should reflect the physics of the same
environment (i.e. same experimental setup) as the forward system, but with any explicit time-dependence of the environment reversed and the time-inversion should reverse velocities. In the next two subsections, we will discuss this in detail for both the underdamped and overdamped systems.

\subsection{Time-Inversion for Langevin-Kramers Dynamics}\label{sec:LK}
First consider a general Hamiltonian system coupled to noise and dissipation:
\begin{align}\label{Hamiltonian_sys}
dx_t=(-\Gamma(t,x_t)\nabla H_t(x_t)+\Pi(t,x_t)\nabla H_t(x_t)+G_t(x_t))dt+\tilde \sigma(t,x_t)\circ dW_t,
\end{align}
where $H_t$ is a  time-dependent Hamiltonian function, $\Pi$ is an antisymmetric matrix, $\Gamma$ is a symmetric, positive semidefinite matrix that results in dissipation, $G$ is an additional non-conservative  force field, and $\tilde \sigma$  are the noise coefficients.   

A canonical splitting of the drift, generalizing the terminology in \cite{Chetrite2008}, is given by
\begin{align}\label{Langevin_splitting}
b_+=-\Gamma_t \nabla H_t,\hspace{2mm} b_-=\Pi\nabla H_t+G_t.
\end{align}
 Note that $b_+$ contains the dissipative component of the dynamics (at least, when $H$ is time-independent), while $b_-$ contains the conservative and external force terms.  

To see what the splitting \req{Langevin_splitting} means physically, we will now specialize to underdamped Langevin-Kramers dynamics i.e. the objects in \req{Hamiltonian_sys} have the form
\begin{align}
\Gamma(t,x)=&\left( \begin{array}{cc}
0&0 \\
0 & \gamma(t,q) \end{array} \right), \hspace{2mm} \tilde\sigma(t,x)=\left( \begin{array}{cc}
0&0 \label{Gamma_def}\\
0 & \sigma(t,q) \end{array} \right),\\
 \Pi(t,x)=&\left( \begin{array}{cc}
0&I \\
-I & 0 \end{array} \right),\hspace{2mm} G(t,x)=(0,\tilde F(t,x)),
\end{align}
with Hamiltonian given by \req{hamiltonian_family}.

For the underdamped system, the splitting \req{Langevin_splitting} becomes
\begin{align}
b_+(t,x)=&\left(0,-\frac{1}{m}\gamma(t,q)(p-\psi(t,q))\right),\\
b_-(t,x)=&\left(\frac{1}{m}(p-\psi(t,q)),\frac{1}{m} \delta^{ij}(p_i-\psi_i(t,q))\nabla_q \psi_j(t,q)-\nabla_qV(t,q)+\tilde F(t,x)\right).\notag
\end{align}

Choosing the   time-reversal to be the standard phase-space involution,
\begin{align}\label{p_reversal} 
\phi(q,p)=(q,-p),
\end{align}
  gives
\begin{align}
&(\phi_*b_+)(t,x)=\left(0,\frac{1}{m}\gamma(t,q)(-p-\psi(t,q))\right),\\
&(\phi_*b_-)(t,x)\notag\\
=&\left(\frac{1}{m}(-p-\psi(t,q)),-\frac{1}{m}  \delta^{ij}(-p_i-\psi_i(t,q))\nabla_q \psi_j(t,q)+\nabla_qV(t,q)-\tilde F(t,q,-p)\right).\notag
\end{align}
The time-reversed dynamics are then 
\begin{align}
dq^\prime_t=&\frac{1}{m}(p_t^\prime+\psi(t^*,q_t^\prime))dt,\label{t_rev_q}\\
d(p_t^\prime)_i=&-\frac{1}{m}\delta^{jk}(\gamma_{ij}(t^*,q_t^\prime)+ \partial_{q^i} \psi_j(t^*,q_t^\prime))(p_t^\prime+\psi(t^*,q_t^\prime))_kdt\label{t_rev_p}\\
&+(-\partial_{q^i}V(t^*,q_t^\prime)+\tilde F_i(t^*,q_t^\prime,-p_t^\prime))dt+\sigma(t^*,q_t^\prime)\circ dW_t.\notag
\end{align}
Note that these equations have the same form as the original system, \req{q_eq}-\req{p_eq}, but the explicit time dependence is reversed, the argument $p$ in the external forcing has an additional minus sign, and the vector potential has its sign reversed.  When $\psi=0$ and $\tilde F$ is independent of $p$, this corresponds to the intuitive notion of reversing explicit time dependence of the environment, but otherwise keeping it the same, and comparing the paths to those in the original system.

However, in the presence of a magnetic field, the correct definition of the time-reversed equation is less clear.  There seems to be two reasonable options; either the environment should be fully time-reversed, reversing the direction of all currents, and hence $\psi(t,q)\to -\psi(t^*,q)$ in the time-reversed SDE,  or one can envision that the same experimental setup is used, and hence leave the sign of $\psi$ unchanged.  The former corresponds to the inversion rule \req{p_reversal}.  The latter appears difficult to treat in full generality in a physically reasonable manner, but becomes more tractable in the presence of special symmetries. 

Specifically, in \cite{PhysRevE.96.012160} it was noted that, for a uniform magnetic field, the equations of motion (without noise or drag) are invariant under 
\begin{align}\label{unif_B_sym}
(q_1,q_2,q_3,p_1,p_2,p_3,t)\to(-q_1,q_2,q_3,p_1,-p_2,-p_3,-t).
\end{align} 
 As we will see, the time-inversion on phase space obtained from \req{unif_B_sym}  does lead to a pair of forward and backward equations that maintain the same background environment, including the external uniform magnetic field, only with explicit time dependence reversed.  In terms of entropy production, the derivation proceeds similarly to that of the standard phase-space involution. Further discussion of this case can be found in Section \ref{sec:unif_B}.

\subsection{Time-Inversion for the Overdamped Limit}\label{sec:time_inv_overdamp}
For the underdamped system, we used the splitting \req{Langevin_splitting} to derive the time-reversed SDEs.  We don't need to independently choose such a splitting for the overdamped system, as the small-mass limit uniquely determines the forward and time-reversed paths.

Together, Theorem \ref{theorem:conv} and \req{q_strat_SDE} give  SDEs for the small-mass limits of the forward and backward processes respectively:
\begin{align}\label{q_forward_eq}
dq_t=&\tilde\gamma^{-1}(t,q_t)\left(-\partial_t\psi(t,q_t)-\nabla_qV(t,q_t)+\tilde F(t,q_t,\psi(t,q_t))\right) dt\\
& +\tilde S(t,q_t)dt+\tilde\gamma^{-1}(t,q_t)\sigma(t,q_t)\circ dW_t,\notag
\end{align}
\begin{align}
dq^\prime_t=&(\tilde\gamma^T)^{-1}(t^*,q^\prime_t)\left(-\partial_t\psi(t^*,q^\prime_t)-\nabla_qV(t^*,q_t^\prime)+ \tilde F(t^*,q^\prime_t,\psi(t^*,q^\prime_t))\right) dt\\
& +\tilde S^\prime(t^*,q^\prime_t)dt+(\tilde\gamma^T)^{-1}(t^*,q^\prime_t)\sigma(t^*,q^\prime_t)\circ dW_t,\notag
\end{align}
where $\tilde S$ and $\tilde S^\prime$ are computed via \req{tilde_S_def} using the vector potentials $\psi$ and $-\psi$ respectively.

The natural configuration-space involution for the overdamped dynamics, inherited from \req{p_reversal}, is simply the identity map, $q\to q$.  To compute the entropy production via the framework of Section \ref{sec:entropy}, we need to show that the SDE for $q_t^\prime$ has the form \req{gen_time_inv_rule}; specifically, we need to find a Wiener process, $\tilde W_t$, that satisfies
\begin{align}\label{tildeW_eq}
(\tilde\gamma^T)^{-1}(t^*,q^\prime_t)\sigma(t^*,q^\prime_t) dW_t=\tilde\gamma^{-1}(t^*,q^\prime_t)\sigma(t^*,q^\prime_t) d\tilde W_t
\end{align}
i.e. we need to show that
\begin{align}\label{tildeW_def}
d\tilde W_t\equiv \left((\tilde\gamma^{-1}\sigma)^{-1}(\tilde\gamma^T)^{-1}\sigma\right)(t^*,q^\prime_t) dW_t
\end{align}
is a Wiener process. 

By Levy's theorem (see p.157 in \cite{karatzas2014brownian}), $\tilde W_t$ is a Wiener process if
\begin{align}\label{Levy_eq}
&B^{-1}AA^T(B^{-1})^T=I,\,\,\,\, A\equiv(\tilde\gamma^T)^{-1}\sigma,\,\,B\equiv \tilde\gamma^{-1}\sigma,
\end{align}
where $I$ is the identity matrix and, for the purpose of employing matrix notation, we define $\tilde\gamma^i_j\equiv\delta^{ik}\tilde\gamma_{kj}$ and $\sigma^{i}_j\equiv\delta^{ik}\sigma_{kj}$.  For arbitrary matrices $\tilde\gamma$ and $\sigma$, \req{Levy_eq} will not always hold.  However, given our assumption of the fluctuation-dissipation relation, \req{fluc_dis}, one finds that proving \req{Levy_eq} is equivalent to proving
\begin{align}
\tilde\gamma(\tilde\gamma^T)^{-1}\gamma\tilde\gamma^{-1}\tilde\gamma^T=\gamma.
\end{align}
This identity can be derived from the relations
\begin{align}
\gamma=\frac{1}{2}(\tilde\gamma+\tilde\gamma^T),\,\,\,\,\,\,\,\tilde \gamma^{-1}+(\tilde\gamma^T)^{-1}=2(\tilde\gamma^T)^{-1}\gamma\tilde\gamma^{-1},
\end{align}
and so $\tilde W_t$ is a Wiener process.

Therefore, in Stratonovich form, the backward SDE can be written
\begin{align}\label{q_backward_eq}
dq^\prime_t=&(\tilde\gamma^T)^{-1}(t^*,q^\prime_t)\left(-\partial_t\psi(t^*,q^\prime_t)-\nabla_qV(t^*,q_t^\prime)+ \tilde F(t^*,q^\prime_t,\psi(t^*,q^\prime_t))\right) dt\\
& +\tilde S(t^*,q_t^\prime)dt+ \left(S^\prime(t^*,q^\prime_t)-S(t^*,q_t^\prime)\right)dt +\tilde\gamma^{-1}(t^*,q^\prime_t)\sigma(t^*,q^\prime_t)\circ d\tilde W_t,\notag
\end{align}
where $S$ and $S^\prime$ are computed via \req{Ito_S} using the vector potentials $\psi$ and $-\psi$ respectively.  

In this form, the forward and backward equations are seen to correspond, as in framework of Section \ref{sec:entropy}, under the splitting
\begin{align}
b_+=&\tilde\gamma^{-1}F+\tilde S-\frac{1}{2}\left(\tilde\gamma^{-1}-(\tilde\gamma^T)^{-1}\right)F-\frac{1}{2}(S-S^\prime),\label{overdamp_b+}\\
b_-=&\frac{1}{2}\left(\tilde\gamma^{-1}-(\tilde\gamma^T)^{-1}\right)F+\frac{1}{2}(S-S^\prime),\label{overdamp_b-}
\end{align}
where 
\begin{align}
&F(t,q)\equiv-\partial_t\psi(t,q)-\nabla_qV(t,q)+\tilde F(t,q,\psi(t,q)),\\
&\frac{1}{2}\left(\tilde\gamma^{-1}-(\tilde\gamma^T)^{-1}\right)^{ij}F_j=  (\tilde\gamma^{-1})^{ki}H_{k\ell}(\tilde\gamma^{-1})^{\ell j}F_j,\\
&\frac{1}{2}(S-S^\prime)^i=\beta^{-1}\partial_{q^j}\left((\tilde\gamma^{-1})^{ki}H_{k\ell}(\tilde\gamma^{-1})^{\ell j}\right).
\end{align}

Having identified the forward and backward equations in both the underdamped and overdamped regimes,  we proceed in Section \ref{sec:entropy_limit} to investigate the entropy production in the environment for the underdamped system and  derive a formula for its small-mass limit.  

Then, in Section \ref{sec:anom_entropy} we compute  the entropy production in the environment for the overdamped system and  compare it to the limit of the entropy production in the  underdamped system.  This allows us to identify the entropy anomaly.

Finally, Section \ref{sec:unif_B} treats a special case of a uniform magnetic field, using the alternative time-inversion obtained from \req{unif_B_sym}. Here, we are also able to compute the  overdamped entropy production  and isolate the entropy anomaly.

\section{Entropy Production for Underdamped Langevin-Kramers Dynamics}\label{sec:entropy_limit}
In this section we compute the entropy production in the underdamped system, \req{q_eq}-\req{p_eq}, that results from the splitting \req{Langevin_splitting}.  

Using \req{S_env_def}, along with the assumption that the noise only couples to the momentum, we find
\begin{align}
S^{env}_{s,t}= &-\int_{s}^t  2\partial_{p_l} H (r,x_r)  \gamma_{jl}(r,q_r)(\Sigma^{-1})^{jk}(r,q_r)\circ d(p_r)_k\notag\\
&+\int_{s}^t  2\partial_{p_l} H (r,x_r)  \gamma_{jl}(r,q_r)(\Sigma^{-1})^{jk}(r,q_r)(- \nabla_q H+\tilde F)_k(r,x_r)  -\nabla_p\cdot \tilde F(r,x_r)dr.\notag
\end{align}
The fluctuation-dissipation relation, \req{fluc_dis}, yields
\begin{align}
S^{env}_{s,t}= &-\int_{s}^t  \beta(r,q_r) \partial_{p_k} H (r,x_r) \circ d(p_r)_k\notag\\
&+\int_{s}^t  \beta(r,x_r) \partial_{p_k} H (r,x_r)(- \nabla_q H+\tilde F)_k(r,x_r)  -\nabla_p\cdot \tilde F(r,x_r)dr\notag.
\end{align}
The fact that $H\in C^{1,3}$, $\nabla_pH\in C^{1,2}$, where the first index refers to the $(t,q)$-variables and the second to the $p$-variables, allows us to use It\^o's formula for the Stratonovich integral to obtain
\begin{align}
&\beta(t,q_t)H(t,x_t)-\beta(s,q_s)H(s,x_s)\\
=&\int_s^t \partial_r(\beta H)(r,x_r)dr+\int_s^t \nabla_q (\beta H)(r,x_r) \cdot dq_r+\int_s^t \nabla_p (\beta H)(r,x_r) \circ dp_r.\notag
\end{align}

Therefore
\begin{align}
S^{env}_{s,t}= &-\left(\beta(t,q_t)H(t,x_t)-\beta(s,q_s)H(s,x_s)\right)+\int_s^t \partial_r(\beta H)_rdr\\
&+\int_s^t H_r \nabla_q \beta_r \cdot \nabla_p H_r  dr+\int_s^t  \beta_r \nabla_q H_r\cdot \nabla_p H_r dr\notag\\
&+\int_{s}^t  \beta_r\nabla_p H_r\cdot (- \nabla_q H_r+\tilde F_r)  -\nabla_p\cdot \tilde F_rdr\notag\\
= &-\left(\beta(t,q_t)H(t,x_t)-\beta(s,q_s)H(s,x_s)\right)+\int_s^t \partial_r(\beta H)_rdr\notag\\
&+\int_s^t \left( H_r \nabla_q \beta_r+ \beta_r\tilde F_r \right)\cdot \nabla_p H_r-\nabla_p\cdot \tilde F_r  dr\notag.
\end{align}

Next, we use the form of the Hamiltonian, \req{hamiltonian_family}, along with an additional assumption:
\begin{assumption}\label{assump:F_p_ind}
For the remainder of this work, we assume $\tilde F$ is independent of $p$.
\end{assumption}
Recalling $z_t^m=u_t^m/\sqrt{m}$, the entropy produced in the environment arising from the splitting \req{Langevin_splitting} can be written 
\begin{align}\label{S_env_final}
&S^{env,m}_{s,t}\\
= &-\left(\beta(t,q^m_t)H(t,q^m_t,z^m_t)-\beta(s,q^m_s)H(s,q^m_s,z^m_s)\right)+\int_s^t  \partial_r (\beta V)(r,q_r^m)dr\notag\\
&+\frac{1}{2}\int_s^t \partial_r \beta(r,q_r^m)  \|z_r^m\|^2dr+\frac{1}{2\sqrt{m}}\int_s^t \|z_r^m\|^2  \nabla_q \beta(r,q_r^m)\cdot z^m_r  dr\notag\\
&+\frac{1}{\sqrt{m}}\int_s^t \left((V\nabla_q \beta)(r,q_r^m) +(\beta\tilde F)(r,q_r^m)-  (\beta \partial_r\psi)(r,q_r^m)\right)  \cdot z^m_r dr,\notag
\end{align}
where $H(t,q,z)\equiv\frac{1}{2}\|z\|^2+V(t,q)$.

\subsection{Small-mass Limit of the Underdamped Entropy Production}\label{sec:underdamped_entropy_limit}
Having obtained a formula, \req{S_env_final}, for the entropy produced in the environment, we now investigate its small-mass limit. The terms in \req{S_env_final} of the form $\beta H$ converge in distribution by Theorem \ref{theorem:conv_dist}.  The term $\int_s^t  \partial_r (\beta V)(r,q_r^m)dr$ will be shown to converge by using Theorem \ref{theorem:conv} (i.e. because $q_t^m\rightarrow q_t$). 

What remains are sums of integral processes of the forms
\begin{align}
\int_s^t G(r,q_r^m,z^m_r)dr \text{ and } m^{-1/2}\int_s^t z^m_r\cdot  K(r,q_r^m,z^m_r)dr,
\end{align}
 where $G$ and $K$ are multi-linear in $z$ and $K$ is even in $z$.  These can be homogenized by the limit formulas summarized below. The proofs, which are somewhat long and technical, have been collected in Appendix \ref{sec:homog}. Specifically, see  Assumption \ref{assump:beta_deriv} and Theorems \ref{homog_thm_1} and \ref{homog_thm_2} for a precise listing of the required properties and the resulting modes of convergence.  For the remainder of this paper, we will work under Assumption \ref{assump:beta_deriv}.

To state the convergence results, first define the Gibbs distribution, pointwise in $(t,q)$,
\begin{align}\label{h_def}
h(t,q,z)=\left(\frac{\beta(t,q)}{2\pi}\right)^{n/2}e^{-\beta(t,q)\|z\|^2/2}.
\end{align}
Given functions $B^{i_1,...,i_k}(t,q)$,  consider the family of processes
\begin{align}
J_{s,t}^m\equiv\int_s^t B^{i_1,...,i_k}(r,q_r^m) (z_r^m)_{i_1}...(z_r^m)_{i_k} dr,
\end{align}
where  $0\leq s\leq t\leq T$.  The first homogenization formula is
\begin{align}\label{J_limit1}
\lim_{m\to 0}J_{s,t}^m=\int_s^t B^{i_1,...,i_k}(r,q_r)  \left(\int h(r,q_r,z) z_{i_1}...z_{i_k} dz \right) dr\equiv J_{s,t}.
\end{align}
The convergence is in a sup-$L^p$ norm and the convergence rate is $O(m^{1/2})$; see Theorem \ref{homog_thm_1} for details. Intuitively, the limit is obtained by averaging the fast scaled-velocity degrees of freedom  over the pointwise equilibrium distribution in $z$. Also note that the limiting process involves only $q_t$, the solution to the overdamped equation, \req{q_SDE}.

When $k$ is odd, the equilibrium average in \req{J_limit1} is zero and we obtain the following leading order result in Theorem \ref{homog_thm_2}:
\begin{align}
&\frac{1}{\sqrt{m}}E\left[J_{s,t}^m\right]\\
=&-\int_s^t  E\left[(-\nabla_q V(r,q_r)-\partial_r \psi(r,q_r)+\tilde F(r,q_r))\cdot \left(\int  (\nabla_z\chi)(r,q_r,z) h(r,q_r,z)dz\right)\right]dr\notag\\
&-\int_s^tE\left[  \int \left((\nabla_q\chi)(r,q_r,z )\cdot z\right) h(r,q_r,z)dz \right]dr+O(m^{1/2}),\notag
\end{align}
where   $h$ is given by \req{h_def} and $\chi$ is defined from $B$ as in \req{chi_def}.

Using the above limit formulas, we  derive the following result concerning the small-mass limit of the expected value of the entropy production.
\begin{theorem}\label{S_conv_thm}
Let $\delta\in (0,1/2)$, $0<s\leq t$, and recall the formula for $S^{env,m}_{s,t}$, \req{S_env_final}. Under Assumption \ref{assump:beta_deriv}, we have
\begin{align}\label{E_S_env}
&E[S^{env,m}_{s,t}]\\
= & E\left[(\beta V)(s,q_s)\right]-E\left[(\beta V)(t,q_t)\right]+\frac{n+2}{2} E\left[\ln(\beta(t,q_t)/\beta(s,q_s))\right]\notag\\
&+\int_s^t E\left[ \partial_r (\beta V)(r,q_r)\right]dr-\int_s^tE\left[ (\beta^{-1} \partial_r \beta)(r,q_r) \right]dr\notag\\
&+\int_s^t  E\left[\left(\left(V\nabla_q \beta +\beta\tilde F-  \beta \partial_t\psi\right) \cdot\tilde\gamma^{-1}F\right)(r,q_r)\right]dr\notag\\
&+\int_s^tE\left[\left(\beta^{-1} \partial_{q^i}\left(\left(V\partial_{q^j} \beta +\beta\tilde F_j-  \beta \partial_t\psi_j\right)(\tilde\gamma^{-1})^{ji} \right)\right)(r,q_r)\right]dr\notag\\
&+\int_s^t\!\!E\bigg[\left(\beta^{-3} \partial_{q^{k}}\beta\left(\frac{n}{2}\delta^k_l+\delta^{ki_3}\delta^{i_1i_2} G_{i_1i_2i_3}^{j_1 j_2j_3}\delta_{j_2j_3}\tilde\gamma_{j_1 l}\right.\right.\notag\\
&\hspace{12mm}\left. \left.-  \delta^{ki_3}\delta^{i_1i_2} G_{i_1i_2i_3}^{j_1 j_2j_3}\gamma_{j_1j_2}\delta_{j_3l} \right) (\tilde\gamma^{-1})^{li}\partial_{q^i}\beta \right)(r,q_r)\bigg]dr+O(m^\delta)\notag
\end{align}
as $m\rightarrow 0$, where
\begin{align}\label{G_def}
G_{i_1i_2i_3}^{j_1j_2j_3}=\delta^{j_1k_1}\delta^{j_2k_2}\delta^{j_3k_3}\int_0^\infty (e^{-y\tilde \gamma})_{i_1k_1}  (e^{-y\tilde \gamma})_{i_2k_2} (e^{-y\tilde \gamma})_{i_3k_3} dy.
\end{align}
Note that, for the purposes of computing the operator exponential, $\tilde\gamma$ is considered to be the linear map $z_i\to \tilde\gamma_{ij}\delta^{jk}z_k$.
\end{theorem}
\begin{proof}

Fix $\delta\in (0,1/2)$ and $0<s\leq t$. From \req{S_env_final} we see
\begin{align}
&E[S^{env,m}_{s,t}]\\
= &\frac{1}{2}E\left[\beta(s,q^m_s)\|z_s^m\|^2\right]-\frac{1}{2}E\left[\beta(t,q^m_t)\|z_t^m\|^2\right]\notag\\
&+E\left[\beta(s,q^m_s)V(s,q^m_s)\right]-E\left[\beta(t,q^m_t)V(t,q^m_t)\right]+E\left[\int_s^t  \partial_r (\beta V)(r,q_r^m)dr\right]\notag\\
&+\frac{1}{2}E\left[\int_s^t \partial_r \beta(r,q_r^m)  \|z_r^m\|^2dr\right]+\frac{1}{2\sqrt{m}}E\left[\int_s^t \|z_r^m\|^2  \nabla_q \beta(r,q_r^m)\cdot z^m_r  dr\right]\notag\\
&+\frac{1}{\sqrt{m}}E\left[\int_s^t \left((V\nabla_q \beta)(r,q_r^m) +(\beta\tilde F)(r,q_r^m)-  (\beta \partial_r\psi)(r,q_r^m)\right)  \cdot z^m_r dr\right].\notag
\end{align}

For fixed $t$, $\tilde h(q,z)=\beta(t,q) \|z\|^2$ is $C^1$ in $(q,z)$ and satisfies \req{nabla_h_bound}, therefore Theorem \ref{theorem:conv_dist} gives
\begin{align}
&E\left[\beta(t,q_t^m)\|z_t^m\|^2\right]\\
=&E\left[ \left(\frac{ \beta(t,q_t)}{2\pi}\right)^{n/2} \int  \beta(t,q_t)\|z\|^2 e^{- \beta(t,q_t)\|z\|^2/2}dz\right]+O(m^{\delta})\notag\\
=&\left(\frac{1}{2\pi}\right)^{n/2} \int  \|w\|^2 e^{-\|w\|^2/2}dw+O(m^{\delta})\notag
\end{align}
as $m\rightarrow 0$.  Note that the first term is independent of $t$, and so the first two terms in \req{E_S_env} cancel up to order $m^\delta$.

Our assumptions imply $\beta V$ and $\partial_t(\beta V)$ are $C^1$ with polynomially bounded zeroth and first derivatives in $q$ (uniform in $t\in [0,T]$).  Therefore
 \begin{align}
&\left|E[\beta(t,q_t^m)V(t,q_t^m)]-E[\beta(t,q_t)V(t,q_t)]\right|\\
\leq &E[|\tilde C(1+\|q_t\|^{\tilde p}+\|q_t-q_t^m\|^{\tilde p})\|q_t-q_t^m\||]\notag\\
=&O(m^{1/2})\notag
\end{align}
and similarly,
 \begin{align}
\left|E\left[\int_s^t  \partial_r (\beta V)(r,q_r^m)dr\right]-E\left[\int_s^t  \partial_r (\beta V)(r,q_r)dr\right]\right|=O(m^{1/2}).
\end{align}

The integrands involving $z_t^m$ are all multi-linear functions of the $z$ variables, hence they can be handled using the results in Appendix \ref{sec:homog}.  If we consider the first $z$-dependent term we see that  $\partial_t\beta$ is $C^1$ with zeroth and first derivatives that are polynomially bounded in $q$, uniformly in $t\in[0,T]$. Therefore Theorem \ref{homog_thm_1} gives
\begin{align}
&E\left[\int_s^t \partial_r \beta(r,q_r^m)  \|z_r^m\|^2dr\right]\\
=& E\left[\int_s^t \partial_r \beta(r,q_r)  \int  \|z\|^2 h(r,q_r,z)dr\right]+O(m^{1/2})\notag\\
=& E\left[n\int_s^t \beta^{-1}(r,q_r)\partial_r \beta(r,q_r)dr \right]+O(m^{1/2})\notag.
\end{align}

The last two terms are proportional to $1/\sqrt{m}$. The integrands are rank $3$ and rank $1$ tensors respectively, evaluated on $z_t^m$ and have the required differentiability and polynomial boundedness properties to apply Theorem \ref{homog_thm_2}. Therefore
\begin{align}
&\frac{1}{2\sqrt{m}}E\left[\int_s^t \|z_r^m\|^2  \nabla_q \beta(r,q_r^m)\cdot z^m_r  dr\right]\\
&+\frac{1}{\sqrt{m}}E\left[\int_s^t \left((V\nabla_q \beta)(r,q_r^m) +(\beta\tilde F)(r,q_r^m)-  (\beta \partial_r\psi)(r,q_r^m)\right)  \cdot z^m_r dr\right]\notag\\
=&-\int_s^t  E\left[(-\nabla_q V(r,q_r)-\partial_r\psi(r,q_r)+\tilde F(r,q_r))\cdot \left(\int  (\nabla_z\chi)(r,q_r,z) h(r,q_r,z)dz\right)\right]dr\notag\\
&-\int_s^tE\left[  \int \left((\nabla_q\chi)(r,q_r,z )\cdot z\right) h(r,q_r,z)dz \right]dr+O(m^{1/2}).\notag
\end{align}
   $h$ was defined in \req{h_def} and $\chi=\chi_1+\chi_2$ where $\chi_1$ and $\chi_2$ are computed from
\begin{align}
B_1^{i_1i_2i_3}(t,q)=\frac{1}{2}\delta^{i_1i_2}(\nabla_q \beta)^{i_3}(t,q)
\end{align}
and 
\begin{align}
B_2^{i}(t,q)=\left((V\nabla_q \beta)(t,q) +(\beta\tilde F)(t,q)-  (\beta \partial_t\psi)(t,q)\right)  ^i
\end{align}
 respectively, as described in \req{chi_def} and the surrounding text.

The $\chi_i$ can be computed explicitly by using Lemma \ref{lemma:cell_prob}.
\begin{align}
&\chi_1(z)=-\frac{1}{2}\delta^{i_1i_2}(\nabla_q\beta)^{i_3} G_{i_1i_2i_3}^{j_1j_2j_3} \\
&\times \left(z_{j_1}z_{j_2}z_{j_3}+2\beta^{-1}\left(\gamma_{j_1j_2}\delta_{j_3l}(\tilde\gamma^{-1})^{lk} +\gamma_{j_1j_3}\delta_{j_2l}( \tilde\gamma^{-1} )^{lk}+\gamma_{j_2j_3}\delta_{j_1l}(\tilde\gamma^{-1})^{lk}\right)z_k \right),\notag
\end{align}
where we suppress the $(t,q)$ dependence and define
\begin{align}
G_{i_1i_2i_3}^{j_1j_2j_3}=\delta^{j_1k_1}\delta^{j_2k_2}\delta^{j_3k_3}\int_0^\infty (e^{-y\tilde \gamma})_{i_1k_1}  (e^{-y\tilde \gamma})_{i_2k_2} (e^{-y\tilde \gamma})_{i_3k_3} dy,
\end{align}
and
\begin{align}
\chi_2(z)=&-\left(V\nabla_q \beta+\beta\tilde F-  \beta \partial_t\psi\right)_i (\tilde\gamma^{-1})^{ij}z_j.
\end{align}

Therefore
\begin{align}
&\int  (\partial_{z_l}\chi)(z) h(z)dz\\
=&-\left(V\nabla_q \beta +\beta\tilde F-  \beta \partial_t\psi\right)_i (\tilde\gamma^{-1})^{il}\notag \\
& -\beta^{-1} \delta_{j_1j_2}G_{i_1i_2i_3}^{j_1j_2l} \left ( \delta^{i_1i_3}(\nabla_q\beta)^{i_2}+\frac{1}{2} \delta^{i_1i_2}(\nabla_q\beta)^{i_3}\right)\notag\\
& -\beta^{-1}\delta^{i_1i_2}(\nabla_q\beta)^{i_3} G_{i_1i_2i_3}^{j_1j_2j_3} \left ( \gamma_{j_1j_2}\delta_{j_3k}(\tilde\gamma^{-1})^{kl}+2 \gamma_{j_2j_3}\delta_{j_1k}(\tilde\gamma^{-1})^{kl}\right),\notag
\end{align}
and
\begin{align}
&\int \left((\nabla_q\chi)(z )\cdot z\right) h(z)dz\\
=&-\frac{1}{2}\delta^{i_1i_2}\beta^{-2}\partial_{q^i}\left((\nabla_q\beta)^{i_3} G_{i_1i_2i_3}^{j_1j_2j_3}\right)(\delta_{j_1j_2}\delta_{j_3}^{i}  +2\delta_{j_2j_3}\delta_{j_1}^{i} )\notag\\
&-\delta^{i_1i_2}\beta^{-1}\partial_{q^i}\left(\beta^{-1}(\nabla_q\beta)^{i_3} G_{i_1i_2i_3}^{j_1j_2j_3} ( \tilde\gamma^{-1} )^{li}\left(\gamma_{j_1j_2}\delta_{j_3l} +2\gamma_{j_1j_3}\delta_{j_2l}\right)\right)\notag\\
&-\beta^{-1} \partial_{q^i}\left(\left(V\partial_{q^j} \beta +\beta\tilde F_j-  \beta \partial_t\psi_j\right)(\tilde\gamma^{-1})^{ji} \right).\notag
\end{align}
This proves 
\begin{align}\label{E_S_env_old}
&E[S^{env,m}_{s,t}]\\
= &E\left[(\beta V)(s,q_s)\right]-E\left[(\beta V)(t,q_t)\right]+\int_s^t E\left[ \partial_r (\beta V)(r,q_r)\right]dr\notag\\
&+\frac{n}{2} \int_s^tE\left[ \beta^{-1}(r,q_r)\partial_r \beta(r,q_r) \right]dr\notag\\
&+\int_s^t  E\left[(-\nabla_q V(r,q_r)-\partial_r\psi(r,q_r)+\tilde F(r,q_r))\cdot Y_1(r,q_r)\right]dr\notag\\
&+\int_s^tE\left[ Y_2(r,q_r) \right]dr+O(m^\delta)\notag
\end{align}
as $m\rightarrow 0$, where
\begin{align}\label{Y1_def}
(Y_1)^l\equiv&\left(V\nabla_q \beta +\beta\tilde F-  \beta \partial_t\psi\right)_i (\tilde\gamma^{-1})^{il} \\
&+\beta^{-1} \delta_{j_1j_2}G_{i_1i_2i_3}^{j_1j_2l}  \left( \delta^{i_1i_3}(\nabla_q\beta)^{i_2}+\frac{1}{2} \delta^{i_1i_2}(\nabla_q\beta)^{i_3}\right)\notag\\
& +\beta^{-1}\delta^{i_1i_2}(\nabla_q\beta)^{i_3} G_{i_1i_2i_3}^{j_1j_2j_3}  \left( \gamma_{j_1j_2}\delta_{j_3k}(\tilde\gamma^{-1})^{kl}+2 \gamma_{j_2j_3}\delta_{j_1k}(\tilde\gamma^{-1})^{kl}\right)\notag
\end{align}
and
\begin{align}\label{Y2_def}
Y_2\equiv &\frac{1}{2}\delta^{i_1i_2}\beta^{-2}\partial_{q^i}\left((\nabla_q\beta)^{i_3} G_{i_1i_2i_3}^{j_1j_2j_3}\right)(\delta_{j_1j_2}\delta_{j_3}^{i}  +2\delta_{j_2j_3}\delta_{j_1}^{i} )\\
&+\delta^{i_1i_2}\beta^{-1}\partial_{q^i}\left(\beta^{-1}(\nabla_q\beta)^{i_3} G_{i_1i_2i_3}^{j_1j_2j_3} ( \tilde\gamma^{-1} )^{li}\left(\gamma_{j_1j_2}\delta_{j_3l} +2\gamma_{j_1j_3}\delta_{j_2l}\right)\right)\notag\\
&+\beta^{-1} \partial_{q^i}\left(\left(V\partial_{q^j} \beta +\beta\tilde F_j-  \beta \partial_t\psi_j\right)(\tilde\gamma^{-1})^{ji} \right).\notag
\end{align}

Using It\^o's formula we can compute
\begin{align}\label{E_log_beta}
&E\left[\ln(\beta(t,q_t)/\beta(s,q_s))\right]\\
=&\int_s^tE\left[ (\beta^{-1} \partial_r \beta)(r,q_r) \right]dr+\int_s^tE\left[\left(\beta^{-1}\partial_{q^i}(\beta)(\tilde\gamma^{-1}F+S)^i\right)(r,q_r)\right]dr\notag\\
&+ \int_s^tE\left[\left(\beta^{-1}\partial_{q^i}(\beta^{-1}\partial_{q^j}\beta) (\tilde\gamma^{-1})^{ji}\right)(r,q_r)\right]dr\notag\\
&+\int_s^t E\left[\left(\beta^{-1}\partial_{q^i}(\beta^{-1}\partial_{q^j}\beta) (\tilde\gamma^{-1})^{ik}H_{k\ell}(\tilde\gamma^{-1})^{j\ell}\right)(r,q_r)\right]dr.\notag
\end{align}
Note the the last term vanishes by antisymmetry of $H$ combined with symmetry of the derivative terms.

$Y_1$ and $Y_2$ can be simplified using the identities
\begin{align}
 &\frac{1}{2} \delta_{j_1j_2}G_{i_1i_2i_3}^{j_1j_2\eta} \delta^{i_1i_2}\tilde\gamma_{\eta k}+\delta^{i_1i_2}G_{i_1i_2i_3}^{j_1j_2j_3} \gamma_{j_1j_2}\delta_{j_3k}\\
 =&-\frac{1}{2}\int_0^\infty \frac{d}{dy}\left[\sum_{\alpha,\eta}(e^{-y\tilde\gamma})_{\alpha\eta}(e^{-y\tilde\gamma})_{\alpha\eta}  (e^{-y\tilde\gamma})_{i_3 k}\right]dy=\frac{n}{2} \delta_{i_3 k},\notag
\end{align}
\begin{align}
&\delta_{j_1j_2}G_{i_1i_3\alpha}^{j_1j_2\eta} \delta^{i_1\alpha}\tilde\gamma_{\eta k} +2\delta^{i_1i_2}G_{i_1i_2i_3}^{j_1j_2j_3}  \gamma_{j_2j_3}\delta_{j_1k}\\
=&-\int_0^\infty\frac{d}{dy}\left[\sum_{\alpha,\eta}(e^{-y\tilde\gamma})_{\alpha k}(e^{-y\tilde \gamma})_{\alpha\eta}(e^{-y\tilde\gamma})_{i_3\eta}\right]dy=\delta_{i_3k},\notag
\end{align}
and, similarly,
\begin{align}
&\delta^{i_1i_2} G_{i_1i_2i_3}^{\eta j_2j_3}\delta_{j_2j_3}\tilde\gamma_{\eta k}  +2\delta^{i_1i_2} G_{i_1i_2i_3}^{j_1j_2j_3} \gamma_{j_1j_3}\delta_{j_2k}=\delta_{i_3k}.
\end{align}

These yield
\begin{align}\label{Y1_simp}
(Y_1)^l=&\left(V\nabla_q \beta +\beta\tilde F-  \beta \partial_t\psi\right)_i (\tilde\gamma^{-1})^{il} +\frac{n+2}{2}\beta^{-1}\partial_{q^{i}}(\beta)  (\tilde\gamma^{-1})^{il} 
\end{align}
and
\begin{align}\label{Y2_simp}
Y_2= &\frac{n+2}{2}\beta^{-2}\partial_{q^i}\left[ \partial_{q^{j}}(\beta)(\tilde\gamma^{-1})^{ji}\right]+\beta^{-1} \partial_{q^i}\left(\left(V\partial_{q^j} \beta +\beta\tilde F_j-  \beta \partial_t\psi_j\right)(\tilde\gamma^{-1})^{ji} \right)\\
&-\beta^{-3}\partial_{q^i}\left(\beta\right)(\nabla_q\beta)^{i_3}  ( \tilde\gamma^{-1} )^{li}\left(\delta^{i_1i_2}G_{i_1i_2i_3}^{j_1j_2j_3}\gamma_{j_1j_2}\delta_{j_3l} +2\delta^{i_1i_2}G_{i_1i_2i_3}^{j_1j_2j_3}\gamma_{j_1j_3}\delta_{j_2l}\right).\notag
\end{align}

The final result, \req{E_S_env}, is obtained by combining \req{E_S_env_old}, \req{E_log_beta}, \req{Y1_simp}, and \req{Y2_simp}, after some cancellation and rearrangement.

\end{proof}

In particular, when the vector potential vanishes we have the simplified result:
\begin{corollary}\label{zero_psi_corollary}
Suppose $\psi=0$ (and hence $\tilde\gamma=\gamma$).  Then, for any $0<\delta<1/2$:
 \begin{align}\label{E_S_env2}
&E[S^{env,m}_{s,t}]\\
= &E\left[(\beta V)(s,q_s)\right]-E\left[(\beta V)(t,q_t)\right]+\frac{n+2}{2} E\left[\ln(\beta(t,q_t)/\beta(s,q_s))\right]\notag\\
&+\int_s^t E\left[ \partial_r (\beta V)(r,q_r)\right]dr-\int_s^t E\left[ \left(\beta^{-1}\partial_r \beta\right)(r,q_r) \right]dr\notag\\
&+\int_s^tE\left[\beta^{-1}(r,q_r)\nabla_q\cdot \left(\gamma^{-1}\left(\beta\tilde F+V\nabla_q \beta \right) \right)(r,q_r)  \right]dr\notag\\
&+\int_s^t  E\left[\left(\left(\beta\tilde F+V\nabla_q \beta \right)\cdot \gamma^{-1}\left(-\nabla_q V+\tilde F\right) \right)(r,q_r)\right]dr\notag\\
&+\int_s^tE\left[\left( \beta^{-3}\nabla_q\beta\cdot \left( \frac{3n+2}{6}- \int_0^\infty Tr[\gamma e^{-2y \gamma}] e^{-y \gamma} dy\right)\gamma^{-1}\nabla_q\beta\right)(r,q_r)  \right]dr\notag\\
&+O(m^\delta)\notag
\end{align}
as $m\rightarrow 0$.

\end{corollary}

We obtain further simplification in the case of scalar $\gamma$ (i.e. real-number valued, rather than matrix-valued).
\begin{corollary}
Suppose $\psi=0$ and $\gamma$ is a scalar (still depending on $(t,q)$).  Then
 \begin{align}\label{entropy_anomaly}
&\beta^{-3}\nabla_q\beta\cdot \left( \frac{3n+2}{6}- \int_0^\infty Tr[\gamma e^{-2y \gamma}]e^{-y \gamma} dy\right)\gamma^{-1}\nabla_q\beta\\
=&   \frac{n+2}{6} \beta^{-3}\gamma^{-1}\|\nabla_q\beta\|^2.\notag
\end{align}

\end{corollary}

\section{Overdamped Entropy Production and Entropy Anomaly}\label{sec:anom_entropy}

Now we derive a formula for the entropy produced in the environment for the overdamped system, using the forward and backward equations \req{q_forward_eq} and \req{q_backward_eq}, and compare it to the small-mass limit of the underdamped result from Theorem \ref{S_conv_thm} in order to identify the entropy anomaly.

\subsection{Entropy Production for Overdamped Langevin-Kramers Dynamics}\label{sec:overdamp_entropy}

 We  need the following assumption on $\nabla_q V$ to ensure that the formula for the entropy production is well-defined:
\begin{assumption}\label{assump:nablaV}
 $\nabla_qV$ is $C^2$.
\end{assumption}
With this, we can apply \req{S_env_def} to the splitting \req{overdamp_b+}-\req{overdamp_b-} to get
\begin{align}\label{S_env_0_eq}
S^{env,0}_{s,t}=&\int_{s}^t 2\hat b_+^j(r,q_r)  (\tilde \Sigma^{-1})_{jk}(r,q_r)\circ dq^k_r\\
&-\int_{s}^t  2\hat b_+^j(r,q_r) (\tilde \Sigma^{-1})_{jk}(r,q_r)b^k_-(r,q_r)  +\nabla\cdot b_-(r,q_r)dr,\notag
\end{align}
where
\begin{align}
\hat b_+^i=& (\tilde\gamma^{-1})^{ij}F_j-(\tilde\gamma^{-1})^{ik}\partial_k\beta^{-1}+(\partial_j(\beta^{-1})-F_j)(\tilde\gamma^{-1})^{ik}H_{ k\ell}(\tilde\gamma^{-1})^{j\ell},\\
b_-^i=&(\tilde\gamma^{-1})^{k i}H_{k\ell}(\tilde\gamma^{-1})^{\ell j}F_j+\beta^{-1}\partial_{q^j}\left((\tilde\gamma^{-1})^{k i}H_{k\ell}(\tilde\gamma^{-1})^{\ell j}\right),\notag\\
\tilde \Sigma^{ij}=&2\beta^{-1}(\tilde\gamma^{-1})^{ik}\gamma_{k\ell}(\tilde\gamma^{-1})^{j\ell}.\notag
\end{align}
Using It\^o's formula and \req{q_SDE} to compute the first term in \req{S_env_0_eq}, and then taking expected values, we obtain, after substantial cancellation:
\begin{align}\label{E_S_env_0}
&E[S^{env,0}_{s,t}]\\
=&E[(\beta V)(s,q_s)]-E[(\beta V)(t,q_t)]+\int_s^tE\left[\partial_r(\beta V)(r,q_r)\right]dr\notag\\
&+E\left[\ln(\beta(t,q_t)/\beta(s,q_s))\right]-\int_s^t E\left[(\beta^{-1}\partial_r\beta)(r,q_r)\right]dr\notag\\
&+\int_{s}^t E\left[\left((V\partial_{q^j}\beta+\beta\tilde F_j -\beta\partial_t\psi_j)(\tilde\gamma^{-1})^{jk}F_k\right)(r,q_r)\right] dr\notag\\
&+\int_{s}^t E\left[\left(\beta^{-1} \partial_{q^i}\left((V\partial_{q^j}\beta+\beta\tilde F_j -\beta\partial_t\psi_j)(\tilde\gamma^{-1})^{j i}\right)\right)(r,q_r)\right]dr.\notag
\end{align}

Combining \req{E_S_env} with \req{E_S_env_0} results in the following relation between the under and overdamped entropy production in the environment
\begin{align}\label{delta_E_S}
&E[S^{env,m}_{s,t}]=E\left[S^{env,0}_{s,t}\right]+\frac{n}{2} E\left[\ln(\beta(t,q_t)/\beta(s,q_s))\right]\\
&+\int_s^tE\bigg[\left(\beta^{-3} \partial_{q^{k}}\beta\left(\frac{n}{2}\delta^k_l+\delta^{ki_3}\delta^{i_1i_2} G_{i_1i_2i_3}^{j_1 j_2j_3}\delta_{j_2j_3}\tilde\gamma_{j_1 l}\right.\right.\notag\\
&\hspace{13mm}\left. \left.-  \delta^{ki_3}\delta^{i_1i_2} G_{i_1i_2i_3}^{j_1 j_2j_3}\gamma_{j_1j_2}\delta_{j_3l} \right) (\tilde\gamma^{-1})^{li}\partial_{q^i}\beta \right)(r,q_r)\bigg]dr+O(m^\delta)\notag
\end{align}
 for any $0<\delta<1/2$.

 When $\psi=0$ we can simplify further to obtain:
 \begin{align}\label{delta_E_S_psi_0}
&E[S^{env,m}_{s,t}]=E\left[S^{env,0}_{s,t}\right]+\frac{n}{2} E\left[\ln(\beta(t,q_t)/\beta(s,q_s))\right]\\
&+\int_s^tE\left[\left( \beta^{-3}\nabla_q\beta\cdot \left( \frac{3n+2}{6}- \int_0^\infty Tr[\gamma e^{-2y \gamma}] e^{-y \gamma} dy\right)\gamma^{-1}\nabla_q\beta\right)(r,q_r)  \right]dr\notag\\
&+O(m^\delta)\notag
\end{align}
for any $0<\delta<1/2$.

\subsection{Definition of the Anomalous Entropy Production}\label{sec:anomaly_def}
In this section we perform a formal calculation that motivates the definition of the anomalous entropy production.

In addition to the entropy produced in the environment, \req{S_env_def}, the diffusing particles also produce  entropy  \cite{PhysRevLett.95.040602,Chetrite2008,0034-4885-75-12-126001,gawedzki2013fluctuation}, defined by
\begin{align}\label{part_entropy_def}
E\left[S^{part}_{s,t}\right]\equiv&- E\left[ \ln(p(t,x_t))\right]+E\left[\ln(p(s,x_s))\right]\\
=&- \int \ln(p(t,x)) p(t,x)dx+\int \ln(p(s,x)) p(s,x)dx,\notag
\end{align}
where $p(t,x)$ is the density of the distribution of $x_t$ with respect to Lebesgue measure. Introduced in \cite{PhysRevLett.95.040602}, the notion of change in stochastic entropy along a particular particle path is debated in the literature. However, seeing the second line of \req{part_entropy_def}, one can also view this not as the expectation of a pathwise quantity, but rather as the change in entropy of the particle's probability distribution from time $s$ to time $t$; the notion of the entropy of a probability distribution is a much more established concept than the pathwise definition. 

Based on the convergence in distribution result, Theorem \ref{theorem:conv_dist}, one expects that the  density, $p^m$, of the underdamped system in the variables $(q,z)$ satisfies
\begin{align}
p^m(t,q,z)=&\left(\frac{ \beta(t,q)}{2\pi}\right)^{n/2}e^{-\beta(t,q)\|z\|^2/2} p^0(t,q)+o(1),
\end{align}
where $p^0(t,q)$ is the density of the overdamped solution, $q_t$.

Therefore, using \req{part_entropy_def} on both the over and underdamped systems, we formally obtain the relation
\begin{align}\label{S_part_formal}
&E\left[S^{part,m}_{s,t}\right]=E\left[S^{part,0}_{s,t}\right]-\frac{n}{2}E\left[ \ln(\beta(t,q_t)/\beta(s,q_s))\right] +o(1).
\end{align}
Physically, in passing to the overdamped limit one has lost (or averaged out) half of the original degrees of freedom.  The entropy in the $z$ degrees of freedom, which in the small-mass limit are locally in equilibrium, can be thought of as the source of the logarithm term in \req{delta_E_S}.

Using \req{S_part_formal} to compare the total entropy production
\begin{align}
E\left[S^{tot,m}_{s,t}\right]\equiv E\left[S^{env,m}_{s,t}\right]+E\left[S^{part,m}_{s,t}\right]
\end{align}
with $E\left[S^{tot,0}_{s,t}\right]$, we obtain another formal relation
\begin{align}
E\left[S^{tot,m}_{s,t}\right]-E\left[S^{tot,0}_{s,t}\right]=E\left[S^{env,m}_{s,t}\right]-E\left[S^{env,0}_{s,t}\right]-\frac{n}{2}E\left[ \ln(\beta(t,q_t)/\beta(s,q_s))\right].
\end{align}
This motivates the definition of the expected anomalous entropy production (or entropy anomaly):
\begin{align}\label{S_anom_def}
&E\left[S_{s,t}^{anom}\right]\equiv\lim_{m\to 0}E\left[S^{env,m}_{s,t}\right]-E\left[S^{env,0}_{s,t}\right]-\frac{n}{2}E\left[ \ln(\beta(t,q_t)/\beta(s,q_s))\right].
\end{align}

\subsection{Entropy Anomaly}\label{sec:anomaly}
Taking the $m\to 0$ limit of \req{delta_E_S} and using the definition \req{S_anom_def} yields  a formula for the entropy anomaly:
\begin{theorem}\label{thm:S_anom}
Under Assumptions \ref{assump:beta_deriv} and \ref{assump:nablaV}, the entropy anomaly, as defined in \req{S_anom_def}, is given by
\begin{align}\label{S_anom}
&E\left[S_{s,t}^{anom}\right]\\
=&\int_s^tE\bigg[\left(\beta^{-3} \partial_{q^{k}}\beta\left(\frac{n}{2}\delta^k_l+\delta^{ki_3}\delta^{i_1i_2} G_{i_1i_2i_3}^{j_1 j_2j_3}\delta_{j_2j_3}\tilde\gamma_{j_1 l}\right.\right.\notag\\
&\hspace{9mm}\left. \left.-  \delta^{ki_3}\delta^{i_1i_2} G_{i_1i_2i_3}^{j_1 j_2j_3}\gamma_{j_1j_2}\delta_{j_3l} \right) (\tilde\gamma^{-1})^{li}\partial_{q^i}\beta \right)(r,q_r)\bigg]dr.\notag
\end{align}
Recall that $G^{j_1j_2j_3}_{i_1i_2i_3}$  was defined in \req{G_def}.
\end{theorem}
This is a new result when $\psi\neq 0$. The case $\psi=0$ has been previously studied by various authors.  We end this section by comparing our result with theirs.  

When $\psi=0$, \req{S_anom} can be simplified to
\begin{align}\label{S_anom_psi_0}
&E\left[S_{s,t}^{anom}\right]\\
=&\int_s^tE\left[\left( \beta^{-3}\nabla_q\beta\cdot \left( \frac{3n+2}{6}- \int_0^\infty Tr[\gamma e^{-2y \gamma}] e^{-y \gamma} dy\right)\gamma^{-1}\nabla_q\beta\right)(r,q_r)  \right]dr.\notag
\end{align}
In particular, for scalar $\gamma$ the entropy anomaly is generated by the term \req{entropy_anomaly}, which matches  Eq. (10) in \cite{PhysRevLett.109.260603}.

More generally, when $\psi=0$ but $\gamma$ is matrix-valued,  one can diagonalize $\gamma=U\Lambda U^{T}$ where $\Lambda$ has diagonal entries $\lambda_i$.  This lets us compute
\begin{align}
&\left(U^T\left(\frac{3n+2}{6}- \int_0^\infty Tr[\gamma e^{-2y \gamma}] e^{-y \gamma} dy\right)U\right)^{ij}
=\left(\frac{1}{3}+\frac{\lambda_i}{2}\sum_l  \frac{1}{2\lambda_l+\lambda_i} \right) \delta^{ij}.
\end{align}
Therefore the entropy anomaly when $\psi=0$  can be equivalently written as
\begin{align}\label{S_anom_compare}
&E\left[S_{s,t}^{anom}\right]\\
=&k_B\int_s^tE\left[\left( \frac{1}{2T}\nabla_qT\cdot \left( \frac{2}{3}\gamma^{-1}+\sum_i(\gamma+2\lambda_i I)^{-1}\right)\nabla_qT\right)(r,q_r)  \right]dr,\notag
\end{align}
where $\lambda_i(t,q)$ are the eigenvalues of $\gamma(t,q)$. Note that the matrix in parentheses in \req{S_anom_compare} is positive definite, hence this formula proves that the entropy anomaly is non-negative when $\psi=0$.

    The physical domains covered by this paper and \cite{PhysRevE.93.012132} do differ, as the latter considers particles with both  translational  and rotational degrees of freedom.  However, in the absence of rotation, \req{S_anom_compare} agrees with the corresponding result in \cite{PhysRevE.93.012132}, Eq. (51c).

\section{Uniform Magnetic Field Case}\label{sec:unif_B}
Finally, in this section, we explore the consequences of using a time-reversal operation other than the standard involution, \req{p_reversal}, for a system in a uniform magnetic field.

Motivated by our previous discussion of the symmetry operation \req{unif_B_sym} from \cite{PhysRevE.96.012160}, we define the time-reversal, $\tilde\phi$, on phase space $(q,p)\in\mathbb{R}^{3}\times\mathbb{R}^{3}$,  with the action
\begin{align}\label{sym_reversal}
(q,p)\to(-q^1,q^2,q^3,p_1,-p_2,-p_3).
\end{align}

In part, the following assumption ensures that the system parameters are compatible with the reversal operation, \req{sym_reversal}.
\begin{assumption}\label{assump:symmetry}
In this section, we make the additional assumptions:
\begin{enumerate}
\item $\gamma$ and $\sigma$ are scalars,
\item $\gamma$ is independent of $q$,
\item $\nabla_qV$ is $C^2$,
\item  $\sigma$ and $V$ are invariant under $q^1\to -q^1$,
\item $\tilde F(t,-q^1,q^2,q^3)=(-\tilde F_1(t,q^1,q^2,q^3),\tilde F_2(t,q^1,q^2,q^3),\tilde F_3(t,q^1,q^2,q^3))$,
\item $\psi(t,q^1,q^2,q^3)=\frac{B_0}{2}(-q^2,q^1,0)$.
\end{enumerate}
\end{assumption}
Note that the above vector potential results in a uniform magnetic field of strength $B_0$, pointing in the $\hat{e}_3$ coordinate direction.

With the time-inversion \req{sym_reversal},  the splitting \req{Langevin_splitting}, and also the above assumption, the time-reversed dynamics are
\begin{align}
dq_t^\prime=&\frac{1}{ m}(p_t^\prime-\psi(q_t^\prime))dt,\label{q_rev_sym}\\
d(p_t^\prime)_i=&\left(-\frac{1}{ m}\gamma(t^*)((p_t^\prime)_i-\psi_i(q_t^\prime))-\partial_{q^i} V(t^*,q_t^\prime)+\tilde F(t^*,q_t^\prime)\right.\label{p_rev_sym}\\
&\left.+\frac{1}{ m}\partial_{q^i}\psi_k(q_t^\prime)\delta^{kj}((p_t^\prime)_j-\psi_j(q_t^\prime))\right)dt+\sigma(t^*,q_t^\prime)\delta_{i\rho}dW^\rho_t\notag
\end{align}
i.e.   replace all explicit $t$ dependence with $t^*$. Note that using  \req{gen_time_inv_rule} to obtain \req{p_rev_sym}, one needs to let $\tilde W_t$ be the Wiener process obtained by flipping  the sign of the $p_2$ and $p_3$ components of $W_t$. Recall from the discussion in Section \ref{sec:time_inv} that the time reversed SDE can be defined using any convenient Wiener process, as the choice doesn't impact the distribution of the solution on path space and hence doesn't impact the notion of entropy production.

Unlike \req{t_rev_q}-\req{t_rev_p}, the  SDE \req{q_rev_sym}-\req{p_rev_sym} does reflect the intuition of a time-reversal that maintains the same background environment, including the external uniform magnetic field, only with explicit time dependence reversed.

Next we investigate the overdamped limit.  The natural configuration space involution inherited from \req{sym_reversal}, call it $\hat{\phi}:\mathbb{R}^{3}\to\mathbb{R}^{3}$, is 
\begin{align}\label{unif_B_config_rev}
 (q^1,q^2,q^3)\to (-q^1,q^2,q^3).
\end{align}

 Applying Theorem \ref{theorem:conv} gives the small-mass limit of the forward and backward processes respectively.  We give both the It\^o and Stratonovich forms:
\begin{align}\label{sym_q_forward_eq}
dq_t=&\tilde\gamma^{-1}(t)\left(-\nabla_qV(t,q_t)+\tilde F(t,q_t)\right) dt\\
& +\tilde S(t,q_t)dt+\tilde\gamma^{-1}(t)\sigma(t,q_t)\circ dW_t\notag\\
=&\tilde \gamma^{-1}(t)(-\nabla_q V(t,q_t)+\tilde F(t,q_t))dt+(\tilde\gamma^{-1}\sigma)(t,q_t)dW_t,\notag
\end{align}
\begin{align}\label{sym_q_backward_eq}
dq^\prime_t=&\tilde\gamma^{-1}(t^*)\left(-\nabla_qV(t^*,q^\prime_t)+\tilde F(t^*,q_t)\right) dt\\
& +\tilde S(t^*,q^\prime_t)dt+\tilde\gamma^{-1}(t^*)\sigma(t^*,q^\prime_t)\circ dW_t\notag\\
=&\tilde\gamma^{-1}(t^*)\left(-\nabla_qV(t^*,q^\prime_t)+\tilde F(t^*,q_t)\right) dt+\tilde\gamma^{-1}(t^*)\sigma(t^*,q^\prime_t) dW_t,\notag
\end{align}
where
\begin{align}
\tilde S^i(t,q)=&-\frac{1}{2}\sigma(t,q)\partial_{q^k}\sigma(t,q) (\tilde\gamma^{-1}(t))^{il}\delta_{l\xi}(\tilde\gamma^{-1}(t))^{k\xi}.
\end{align}

Define
\begin{align}
d\tilde W_t\equiv [(\hat{\phi}_*(\tilde\gamma^{-1}\sigma))^{-1}(\tilde\gamma^{-1}\sigma)](t^*,q_t^\prime) dW_t.
\end{align}
Using the formulas
\begin{align}\label{tilde_gamma_matrix}
\tilde \gamma^i_j\equiv\delta^{ik}\tilde\gamma_{kj}=&\left( \begin{array}{ccc}
\gamma&B_0&0 \\
-B_0& \gamma&0\\
0&0&\gamma \end{array} \right)
\end{align}
and
\begin{align}
(\tilde \gamma^{-1})^i_j=(\tilde\gamma^{-1})^{ik}\delta_{kj}=\left( \begin{array}{ccc}
\gamma/(\gamma^2+B_0^2)&-B_0/(\gamma^2+B_0^2)&0 \\
B_0/(\gamma^2+B_0^2)& \gamma/(\gamma^2+B_0^2)&0\\
0&0&\gamma^{-1} \end{array} \right)
\end{align}
we see that
\begin{align}
(\hat{\phi}_*(\tilde\gamma^{-1}\sigma))^{-1}\tilde\gamma^{-1}\sigma)^{i}_\rho\delta^{\rho\eta}(\hat{\phi}_*(\tilde\gamma^{-1}\sigma))^{-1}\tilde\gamma^{-1}\sigma)^{j}_\eta=\delta^{ij}.
\end{align}
Therefore,  Levy's theorem (see p.157 in \cite{karatzas2014brownian}) implies $\tilde W_t$ is a Wiener process.

After rewriting \req{sym_q_backward_eq} as an It\^o SDE driven by $\tilde W_t$ and then converting to Stratonovich form, the backward SDE becomes
\begin{align}\label{back_SDE_2}
dq^\prime_t=&\tilde\gamma^{-1}(t^*)\left(-\nabla_qV(t^*,q^\prime_t)+\tilde F(t^*,q_t)\right) dt\\
&+\tilde S(t^*,q^\prime_t)dt+\hat{\phi}_*(\tilde\gamma^{-1}\sigma)(t^*,q^\prime_t) \circ d\tilde W_t.\notag
\end{align}

 The SDEs \req{sym_q_forward_eq} and \req{back_SDE_2} are related by the time-inversion $\hat{\phi}$, as in \req{gen_time_inv_rule}, when the following splitting is used:
\begin{align}\label{sym_B_overdamp_split}
b_+(t,q)=&\left( \begin{array}{c}
\frac{\gamma(t)}{\gamma^2(t)+B_0^2}(-\partial_{q^1}V(t,q)+\tilde F_1(t,q))\\
\frac{\gamma(t)}{\gamma^2(t)+B_0^2}(-\partial_{q^2}V(t,q)+\tilde F_2(t,q))\\
\gamma^{-1}(t)(-\partial_{q^3}V(t,q)+\tilde F_3(t,q))\end{array}\right)+\left( \begin{array}{c}
-\frac{1}{2(\gamma^2(t)+B_0^2)}\sigma(t,q)\partial_{q^1}\sigma(t,q)\\
-\frac{1}{2(\gamma^2(t)+B_0^2)}\sigma(t,q)\partial_{q^2}\sigma(t,q)\\
-\frac{1}{2\gamma^2(t)}\sigma(t,q)\partial_{q^3}\sigma(t,q)\end{array}\right),\\
b_-(t,q)=&\left( \begin{array}{c}
-\frac{B_0}{\gamma^2(t)+B_0^2}(-\partial_{q^2}V(t,q)+\tilde F_2(t,q))\\
\frac{B_0}{\gamma^2(t)+B_0^2}(-\partial_{q^1}V(t,q)+\tilde F_1(t,q))\\
0\end{array}\right)\notag.
\end{align}

\req{S_env_def} lets us compute the entropy produced in the environment for the overdamped system in terms of  the splitting \req{sym_B_overdamp_split}:
\begin{align}\label{S0_unif_B}
E[S^{env,0}_{s,t}]
= &E\left[(\beta V)(s,q_s)\right]-E\left[\beta V(t,q_t)\right]+E\left[\ln\left(\beta(t,q_t)/\beta(s,q_s)\right)\right]\\
&+\int_s^tE\left[\partial_r(\beta V)(r,q_r)\right]dr-\int_s^tE\left[(\beta^{-1}\partial_r\beta)(r,q_r)\right]dr\notag\\
&+\int_{s}^t E\left[(V\nabla_q\beta+\beta\tilde F)(r,q_r) \cdot \tilde \gamma^{-1}(r)(-\nabla_q V(r,q_r)+\tilde F(r,q_r))\right]dr\notag\\
&+\frac{1}{2}\int_s^t E\left[\sigma^2(r,q_r)\partial_{q^i}(V\nabla_q\beta+\beta\tilde F)_k(r,q_r) (\tilde\gamma^{-1}(r))^{ij}(\tilde\gamma^{-1}(r))^{kl}\delta_{jl}\right]dr\notag\\
&-\int_{s}^t  E\bigg[\frac{B_0\beta(r,q_r)}{\gamma^2(r)+B_0^2}\left(\partial_{q^1}\beta^{-1}(r,q_r)(-\partial_{q^2}V(r,q_r)+\tilde F_2(r,q_r))\right.\notag\\
&\hspace{32mm}\left.-\partial_{q^2}\beta^{-1}(r,q_r)(-\partial_{q^1}V(r,q_r)+\tilde F_1(r,q_r))\right)\bigg]dr\notag\\
&-\int_{s}^tE\left[\frac{B_0}{\gamma^2(r)+B_0^2}\left(-\partial_{q^1}\tilde F_2(r,q_r)+\partial_{q^2}\tilde F_1(r,q_r)\right)\right]dr.\notag
\end{align}

Turning to the underdamped entropy production, we obtain the following by applying Theorem \ref{S_conv_thm} (recall that this result doesn't depend on the choice of time-inversion).  Specifically, we start from \req{E_S_env_old}:\\
For $\delta\in (0,1/2)$ we have
\begin{align}\label{S_m_unif_B}
E[S^{env,m}_{s,t}]
= &E\left[(\beta V)(s,q_s)\right]-E\left[(\beta V)(t,q_t)\right]+\int_s^t E\left[ \partial_r (\beta V)(r,q_r)\right]dr\\
&+\frac{n}{2} \int_s^tE\left[ \beta^{-1}(r,q_r)\partial_r \beta(r,q_r) \right]dr\notag\\
&+\int_s^t  E\left[(-\nabla_q V(r,q_r)+\tilde F(r,q_r))\cdot Y_1(r,q_r)\right]dr\notag\\
&+\int_s^tE\left[ Y_2(r,q_r) \right]dr+O(m^\delta)\notag
\end{align}
as $m\rightarrow 0$.  The integrals in the definitions of $Y_1$ and $Y_2$, see \req{G_def}, \req{Y1_def}, and \req{Y2_def}, can be evaluated by using the fact that $\tilde\gamma$ and $\tilde\gamma^T$ commute, and hence
\begin{align}
e^{-y\tilde\gamma}(e^{-y\tilde\gamma})^T=e^{-2\gamma y} I=(e^{-y\tilde\gamma})^Te^{-y\tilde\gamma}.
\end{align}
Here, $\tilde\gamma^i_j$ is the matrix defined by \req{tilde_gamma_matrix}. After simplification, we obtain the formulas
\begin{align}
(Y_1)^l=&\left(V\nabla_q \beta +\beta\tilde F\right)_i (\tilde\gamma^{-1})^{il}+(1+n/2)\beta^{-1}\partial_{q^j}\beta((\tilde\gamma+2\gamma I)^{-1})^j_{k}  \delta^{kl}\\
& +(n+2)\gamma\beta^{-1}\partial_{q^j}\beta  ((\tilde\gamma+2\gamma I)^{-1})^j_{k} (\tilde\gamma^{-1})^{kl}\notag
\end{align}
and
\begin{align}
Y_2=&\frac{n  +2 }{2}\beta^{-2}((\tilde\gamma+2\gamma I)^{-1})^l_{k}\delta^{ki}\partial_{q^i}(\partial_{q^l}\beta)\\
&  +(n+2)\gamma\beta^{-1}((\tilde\gamma+2\gamma I)^{-1})^l_{k}( \tilde\gamma^{-1} )^{ki}\partial_{q^i}\left(\beta^{-1}\partial_{q^{l}}\beta \right)  \notag\\
&+\beta^{-1} \partial_{q^i}\left(V\partial_{q^j} \beta +\beta\tilde F_j\right)(\tilde\gamma^{-1})^{ji},\notag
\end{align}
where 
\begin{align}\label{tilde_gamma_2_def}
((\tilde \gamma+2\gamma I)^{-1})^i_j=\left( \begin{array}{ccc}
3\gamma/(9\gamma^2+B_0^2)&-B_0/(9\gamma^2+B_0^2)&0 \\
B_0/(9\gamma^2+B_0^2)& 3\gamma/(9\gamma^2+B_0^2)&0\\
0&0&1/(3\gamma) \end{array} \right).
\end{align}

Recalling the definition of the entropy anomaly, \req{S_anom_def}, and using It\^o's formula on $\ln(\beta(t,q_t))$, the results  \req{S0_unif_B} and \req{S_m_unif_B} combine (after a long computation) to yield the entropy anomaly
\begin{align}\label{S_anom_unif_B}
&E\left[S_{s,t}^{anom}\right]=\frac{n+2}{2}\int_{s}^t  E\left[\left(\beta^{-3}\nabla_q\beta\cdot(\tilde\gamma+2\gamma I)^{-1}\nabla_q\beta\right) (r,q_r)\right]dr.
\end{align}
As it should, this expression reduces to the $\psi=0$ result, \req{entropy_anomaly}, when $B_0=0$.  The off-diagonal terms from \req{tilde_gamma_2_def} cancel in \req{S_anom_unif_B} due to antisymmetry, but the  magnetic field  still makes a non-zero contribution via the diagonal terms.  The diagonal terms are all positive and so the above formula proves that the entropy anomaly is non-negative.

\section{Conclusion}
We have investigated the entropy production in underdamped Langevin-Kramers dynamics in a temperature gradient and with matrix-valued drag and magnetic field, and compared this with the overdamped limit.  Specifically,  Theorem \ref{S_conv_thm}  provides a rigorous derivation of the small-mass limit of the entropy produced in the environment for the underdamped system, including a bound on the convergence rate.  

Our procedure uses previously derived rigorous convergence results for process paths (Theorem \ref{theorem:conv}) and the joint distribution of position and scaled velocity (Theorem \ref{theorem:conv_dist}), together with the method of homogenizing integral processes developed in Appendices \ref{sec:homog} and \ref{app:lemma}.  These ideas should generalize to entropy production in other stochastic systems and with  time-inversion rules other than those analyzed here, as well as to  further, mathematically similar, observables.

When the magnetic field vanishes and the standard phase-space time-reversal operation is used, the entropy anomaly derived by our methods, \req{S_anom_compare}, matches the formally derived results in   \cite{PhysRevLett.109.260603} and  \cite{PhysRevE.93.012132}.  Our results  generalize this formula to cover a large class of systems with magnetic fields; see \req{S_anom}.  In addition,  we investigated a special  case of a uniform magnetic field using an alternative time-reversal operation.  There,  we were also able to derive a formula for the entropy anomaly, \req{S_anom_unif_B}. Both of these are new results, not covered by the prior studies \cite{PhysRevLett.109.260603,PhysRevE.93.012132}.

\appendix

\numberwithin{assumption}{section}
\section{Material from \cite{BirrellHomogenization,Birrell2018}}\label{app:assump}
\setcounter{assumption}{0}
    \renewcommand{\theassumption}{\Alph{section}\arabic{assumption}}

    \renewcommand{\thelemma}{\Alph{section}\arabic{lemma}}

In this appendix, we give a list of properties  that, as shown in \cite{BirrellHomogenization,Birrell2018}, are sufficient to  guarantee that Theorems \ref{theorem:conv} and \ref{theorem:conv_dist} hold for the solutions to the SDE \req{q_eq}-\req{p_eq}.

Let $\mathcal{F}^W_t$ be the natural filtration of $W_t$ and $\mathcal{C}$ be any sigma sub-algebra of $\mathcal{F}$ that is independent of $\mathcal{F}^W_\infty$. Define $\mathcal{G}^{W,\mathcal{C}}_t\equiv \sigma(\mathcal{F}^W_t\cup\mathcal{C})$ and complete it with respect to $(\mathcal{G}^{W,\mathcal{C}}_\infty,P)$ to form $\overline{\mathcal{G}^{W,\mathcal{C}}_t}$.  Note that $(W_t,\overline{\mathcal{G}^{W,\mathcal{C}}_t})$ is still a Brownian motion on $(\Omega,\overline{\mathcal{G}^{W,\mathcal{C}}}_\infty,P)$ and this space satisfies the usual conditions \cite{karatzas2014brownian}. 

For the result \req{dis_limit}, we relied on the assumption that our filtered probability space is
\begin{align}\label{prob_space_def}
 (\Omega,\mathcal{F},\mathcal{F}_t,P)\equiv (\Omega, \overline{\mathcal{G}^{W,\mathcal{C}}}_\infty,\overline{\mathcal{G}^{W,\mathcal{C}}_t} ,P).
\end{align}

We also need to assume:
\begin{enumerate}
\item  $V$ is $C^2$, $\gamma$ is $C^2$, $\psi$ is $C^3$, and, letting $\alpha$ denote a multi-index, the following are bounded:
\begin{enumerate}
\item $\nabla_qV$,
\item $\partial_{q^\alpha}\psi$ if $1\leq |\alpha|\leq 3$,
\item $\partial_{q^\alpha}\partial_t\psi$ if $0\leq |\alpha|\leq 2$,
\item $\partial_{q^\alpha}\gamma$ if $1\leq |\alpha|\leq 2$,
\item $\partial_{q^\alpha}\partial_t\gamma$ if $0\leq |\alpha|\leq 1$.
 \end{enumerate}
\item There exists $a,b\geq 0$ s.t. $\tilde V(t,q)\equiv a+b\|q\|^2+V(t,q)$ is non-negative.
\item There exist $C>0$ and $M>0$ such that

\begin{align}\label{H_dot_assump}
|\partial_t V(t,q)|\leq M+C\left(\|q\|^2+\tilde V(t,q)\right).
\end{align}

\item $\gamma$ is symmetric with eigenvalues bounded below by some $\lambda>0$.
\item $\Sigma\equiv\sigma\sigma^T$ has eigenvalues bounded below by $\mu>0$.
\item  $\gamma$, $\tilde F$, $\partial_t\psi$, and $\sigma$ are continuous and bounded.
\item The initial conditions satisfy the following:
\begin{enumerate}
\item There exists $C>0$ such that the (random) initial conditions satisfy $ \|u^m_0\|^2 \leq C m$ for all $m>0$ and all $\omega\in\Omega$.
\item Given any $p>0$ we have $E[\|q_0^m\|^p]<\infty$ for all $m>0$, $E[\|q_0\|^p]<\infty$, and\\
 $E[\|q_0^m-q_0\|^p]^{1/p}=O(m^{1/2})$.
\end{enumerate}
\item $\nabla_q V$ and $\tilde  F$ are Lipschitz in $x$ uniformly in $t$.
\item $\sigma$ is Lipschitz in  $(t,q)$.
\end{enumerate}

\section{Homogenization of Integral Processes}\label{sec:homog}
    \renewcommand{\thecorollary}{\Alph{section}\arabic{corollary}}
    \renewcommand{\thetheorem}{\Alph{section}\arabic{theorem}}

In this appendix, we develop the techniques necessary to investigate the entropy production in the underdamped system, \req{S_env_final}, in the limit $m\rightarrow 0$.

General homogenization results about the  $\epsilon\rightarrow 0^+$ limit of integral processes of the form $\int_0^t G(s,x_s^\epsilon,z_s^\epsilon)ds$, where $x_s^\epsilon$ come from solving some  family of Hamiltonian system parametrized by $\epsilon>0$ (analogous to $m$), can be found in \cite{BirrellHomogSDE}.  Here we summarize and expand on the previous technique to derive explicit formulas for the limit in the case where the integrand is multi-linear in $z$, as well as cover processes of the form $m^{-1/2}\int_s^t z^m_r\cdot  K(r,q_r^m,z^m_r)dr$, an important case that was not treated previously.

As a starting point, let $\chi(t,q,z):[0,\infty)\times\mathbb{R}^{n}\times\mathbb{R}^n\rightarrow\mathbb{R}$ be $C^{1,2}$, meaning  $\chi$ is $C^1$ and, for each $t,q$,  $\chi(t,q,z)$ is $C^2$ in $z$ with second derivatives continuous jointly in all variables.

Using the definitions from Section \ref{sec:prev}, define the operator $L$ and its formal adjoint, $L^*$, by
\begin{align}
(L\chi)(t,q,z)=&\frac{1}{2} \Sigma_{kl}(t,q)(\partial_{z_k}\partial_{z_l}\chi)(t,q,z)-\tilde\gamma_{kl}(t,q)\delta^{li}z_i(\partial_{z_k}\chi)(t,q,z),\label{Hamil_L}\\
(L^*h)(t,q,z)=&\partial_{z_k}\bigg(\frac{1}{2} \Sigma_{kl}(t,q)\partial_{z_l}h(t,q,z)+\tilde\gamma_{kl}(t,x)\delta^{li}z_ih(t,q,z)\bigg).\label{Hamil_L_star}
\end{align}

As in \cite{BirrellHomogSDE}, It\^o's formula  can be used to compute
\begin{align}\label{homog_eq}
&\int_s^t(L\chi)(s,q_r^m,z_r^m)dr\\
=&m^{1/2} (R_1^m)_{s,t}+m\left(\chi(t,q_t^m,z_t^m)-\chi(s,q_s^m,z_s^m)+ (R^m_2)_{s,t}\right),\notag
\end{align}
where we define
\begin{align}\label{R1_def}
&(R_1^m)_{s,t}=-\int_s^t(\nabla_q\chi)(r,q_r^m,z_r^m)\cdot z_r^m dr\\
&-\int_s^t (\nabla_z\chi)(r,q_r^m,z_r^m) \cdot \left[ (-\partial_r\psi(r,q_r^m)+\tilde F(r,q_r^m)-\nabla_qV(r,q_r^m))dr+\sigma(r,q^m_r) dW_r\right]\notag
\end{align}
and
\begin{align}\label{R2_def}
 (R^m_2)_{s,t}=&-\int_s^t\partial_r \chi(r,q_r^m,z_r^m)dr.
\end{align}

Our strategy for homogenizing processes of the form  $\int_s^t G(r,q_r^m,z_r^m)dr$ is to find a function $\tilde G(t,q)$ and a $C^{1,2}$ function $\chi(t,q,z)$ such that 
\begin{align}\label{cell_prob}
L\chi=G-\tilde G.
\end{align}
 A problem of this type is termed a cell problem.  It also appears in formal asymptotic methods for solving the backward Kolmogorov equation as a  series in $\sqrt{m}$ (see Chapter 11 in \cite{pavliotis2008multiscale}), as well as in rigorous homogenization results  (see Chapter 18 in \cite{pavliotis2008multiscale}), and so its appearance as a tool here is not too surprising. 
 
  Assuming $\tilde G$ and $\chi$ exist and don't grow too fast in $z$, we will be able to use \req{homog_eq} to prove
\begin{align}\label{homog_goal}
\int_s^tG(r,q_r^m,z_r^m)dr\rightarrow \int_s^t\tilde G(r,q_r)dr
\end{align}
as $m\to 0$. A solution, $h$, to the adjoint problem, $L^*h=0$, with $\int h(t,q,z)dz=1$, gives us a formula for $\tilde G$ as follows. First multiply \req{cell_prob}  by $h$ and integrate by parts. Assuming the boundary terms are negligible, one obtains
\begin{align}
\tilde G(t,q)=\int h(t,q,z)G(t,q,z)dz.
\end{align}

We will be able to make the above formal derivation rigorous under the following  assumptions.
\begin{assumption}\label{assump:beta_deriv}
From this point on, we assume:
\begin{enumerate}
\item The fluctuation-dissipation relation, Assumption \ref{assump:fluc_dis}, holds. 
\item The properties from Appendix \ref{app:assump} hold.
\item $\tilde F$ is independent of $p$.
\item $\nabla_q\beta$ and $\tilde F$ are $C^2$.
\item  For any $T>0$ and multi-index $\alpha$, the following are polynomially bounded in $q$, uniformly in $t\in[0,T]$:
\begin{enumerate}
\item  $\partial_{q^\alpha} \beta$ if $1\leq|\alpha|\leq 3$,
\item $\partial_{q^\alpha}\partial_t \beta$ if  $0\leq|\alpha|\leq 2$,
\item $\partial_t^2 \beta$,
\item $\partial_{q^\alpha}\partial_t^2 \psi$ if  $0\leq|\alpha|\leq 1$,
\item  $\partial_{q^\alpha} \tilde F$ if  $1\leq|\alpha|\leq 2$,
\item $\partial_{q^\alpha} \partial_t\tilde F$ if  $0\leq|\alpha|\leq 1$,
\item $\partial_{q^\alpha} \partial_t V$ if $0\leq|\alpha|\leq 1$,
\item $\partial_{q^\alpha} V$ if $|\alpha|=2$,
\end{enumerate}
i.e. there exists $\tilde C>0$, $\tilde p>0$ such that
\begin{align}
\sup_{t\in[0,T]}|\partial_t\beta(t,q)|\leq \tilde C(1+\|q\|^{\tilde p})
\end{align}
and so on.
\end{enumerate}

\end{assumption}
With this assumption, $L^*h=0$ is solved by the Gibbs distribution (pointwise in $(t,q)$),
\begin{align}\label{h_def2}
h(t,q,z)=\left(\frac{\beta(t,q)}{2\pi}\right)^{n/2}e^{-\beta(t,q)\|z\|^2/2}.
\end{align}

The integral processes we wish to homogenize are sums of multi-linear functions  in $z$ i.e. they are sums of terms of the form
\begin{align}
G(t,q,z)=B^{i_1,...,i_k}(t,q)z_{i_1}...z_{i_k}.
\end{align}
The solution to the cell problem, \req{cell_prob}, for $G$'s of this form in detailed in Appendix \ref{app:lemma}.

Using Lemma  \ref{lemma:cell_prob} we obtain the following general convergence result, which is used to derive entropy homogenization theorems in Section \ref{sec:underdamped_entropy_limit}. As tools, we will primarily employ the Burkholder-Davis-Gundy inequalities,  H\"older's inequality, and Minkowski's inequality for integrals (see, for example, Theorem 3.28 in \cite{karatzas2014brownian} for the former, and Theorems 6.2 and 6.19 in \cite{folland2013real} for the latter two).  In essence, these are all generalizations of the triangle or Cauchy-Schwarz inequalities to (stochastic) integrals and are all used to decompose the norm of the difference between the $m$-dependent process and its purported limit into pieces, each of which we can show is negligible as $m\to 0$.

\begin{theorem}\label{homog_thm_1}
Let Assumption \ref{assump:beta_deriv} hold, $T>0$ and $B(t,q):\mathbb{R}\times\mathbb{R}^n\rightarrow T^k(\mathbb{R}^n)$ (rank $k$ tensors) be $C^1$  and polynomially bounded in $q$ with polynomially bounded first derivatives, all uniformly in $t\in [0,T]$.

  For $0\leq s\leq t\leq T$, consider the family of processes
\begin{align}\label{eq:J_def}
J_{s,t}^m=\int_s^t B^{i_1,...,i_k}(r,q_r^m) (z_r^m)_{i_1}...(z_r^m)_{i_k} dr.
\end{align}
 Define
\begin{align}
J_{s,t}=\int_s^t B^{i_1,...,i_k}(r,q_r)  \left(\int h(r,q_r,z) z_{i_1}...z_{i_k} dz \right) dr,
\end{align}
where $h$ is given by \req{h_def2}. Then for any $p>0$ we have
\begin{align}
\sup_{0\leq s\leq t\leq T}E\left[\left|J_{s,t}^m-J_{s,t}\right|^p\right]^{1/p}=O(m^{1/2})
\end{align}
as $m\rightarrow 0$.
\end{theorem}
\begin{proof}

Lemma \ref{lemma:cell_prob} implies that for each value of $B$, $\tilde\gamma$, and $\beta$ there exists $A_j\in T^{k-2j}(\mathbb{R}^n)$, $j=0,...,\lfloor (k-1)/2\rfloor$ such that
 \begin{align}\label{chi_def}
\chi(z)=\sum_{j=0}^{\lfloor (k-1)/2\rfloor}A_j^{i_1,...,i_{k-2j}}z_{i_1}...z_{i_{k-2j}}
\end{align}
solves
\begin{align}
(L\chi)(z)=B(z,...,z)-\left(\frac{\beta}{2\pi}\right)^{n/2}\int B(\tilde z,...,\tilde z)e^{-\beta\|\tilde z\|^2/2}d\tilde z,
\end{align}
where $L$ is given by \req{Hamil_L}. Considered as functions of $(\beta,\tilde\gamma,B)$, Lemma \ref{lemma:cell_prob} also shows that the $A_j$ are $C^\infty$, linear in $B$, and every derivative with respect to any number of the $\beta$ and $\tilde\gamma$ variables is bounded by $\tilde C\|B\|$ for some $\tilde C>0$ on any open set of the form 
\begin{align}
U_{\epsilon,R}=\{(\beta,\tilde\gamma,B): \beta>\epsilon, \text{ the symmetric part of $\tilde\gamma$ has spectrum in $(\epsilon,R)$}\}, 
\end{align}
where $R>\epsilon>0$.

 Assumptions \ref{assump:fluc_dis} and \ref{assump:beta_deriv} imply that $(\beta(t,q),\tilde\gamma(t,q),B(t,q))$ map $[0,T]\times\mathbb{R}^n$ into a region of the above form. Therefore 
\begin{align}\label{chi_tqz_def}
\chi(t,q,z)\equiv \chi(\beta(t,q),\tilde\gamma(t,q),B(t,q),z)
\end{align}
is $C^{1,2}$ and there exists $\tilde C,\tilde p>0$ such that
\begin{align}\label{chi_bounds}
&\sup_{t\in [0,T]}\max\{|\chi(t,q,z)|,|\partial_t\chi(t,q,z)|,\|\nabla_q\chi(t,q,z)\|,\|\nabla_z\chi(t,q,z)\|\}\\
\leq& \tilde C(1+\|q\|^{\tilde p})(1+\|z\|^k).\notag
\end{align}

The fact that $\chi(t,q,z)$ is $C^{1,2}$ allows us to apply \req{homog_eq} to obtain
\begin{align}\label{Delta_J_formula}
&J^m_{s,t}-\int_s^t B^{i_1,...,i_k}(r,q^m_r)  \left(\int h(r,q^m_r,z) z_{i_1}...z_{i_k} dz \right) dr\\
=&\int_s^t(L\chi)(r,q_r^m,z_r^m)dr\notag\\
=&m^{1/2}\bigg(-\int_s^t(\nabla_q\chi)(r,q_r^m,z_r^m)\cdot z_r^m dr\notag\\
&-\int_s^t (\nabla_z\chi)(r,q_r^m,z_r^m) \cdot \left[ (-\partial_r\psi(r,q_r^m)+\tilde F(r,q_r^m)-\nabla_q V(r,q_r^m))dr+\sigma(r,q^m_r) dW_r\right]\bigg)\notag\\
&+m\left(\chi(t,q_t^m,z_t^m)-\chi(s,q_s^m,z_s^m)-\int_s^t\partial_r \chi(r,q_r^m,z_r^m)dr\right).\notag
\end{align}

Therefore, for any $p\geq 2$, using the Burkholder-Davis-Gundy inequalities, Minkowski's inequality for integrals,  H\"older's inequality, and Assumption \ref{assump:beta_deriv}, and letting the constant $\tilde C$ vary from line to line,   we obtain
\begin{align}\label{Delta_J_bounds1}
&E\left[\left|J^m_{s,t}-\int_s^t B^{i_1,...,i_k}(r,q^m_r)  \left(\int h(r,q^m_r,z) z_{i_1}...z_{i_k} dz \right) dr\right|^p\right]^{1/p}\\
\leq & m^{1/2}\bigg( E\left[\left|\int_s^t(\nabla_q\chi)(r,q_r^m,z_r^m)\cdot z_r^m + (\nabla_z\chi)(r,q_r^m,z_r^m) \cdot(-\partial_r\psi+\tilde F-\nabla_qV)(r,q_r^m)dr\right|^p\right]^{1/p}\notag\\
&\hspace{13mm}+E\left[\left|\int_s^t (\nabla_z\chi)(r,q_r^m,z_r^m) \cdot\sigma(r,q^m_r) dW_r\right|^p\right]^{1/p}\bigg)\notag\\
&+m\bigg(E\left[\left |  \chi(t,q_t^m,z_t^m) \right|^p\right]^{1/p} +E\left[\left | \chi(s,q_s^m,z_s^m) \right|^p\right]^{1/p}\notag\\
&\hspace{13mm}+E\left[\left| \int_s^t\partial_r \chi(r,q_r^m,z_r^m)dr \right|^p\right]^{1/p}\bigg)\notag\\
\leq & \tilde Cm^{1/2}\bigg( \int_s^t E\left[ (1+\|q_r^m\|^{\tilde p})^p(1+\|z_r^m\|^{k+1})^p\right]^{1/p}dr\notag\\
& \hspace{13mm}+ E\left[\left(\int_s^t \|(\nabla_z\chi)(r,q_r^m,z_r^m) \cdot\sigma(r,q^m_r)\|^2dr \right)^{p/2}\right]^{1/p}\bigg)\notag\\
&+\tilde C m\left( E\left[ (1+\|q_t^m\|^{\tilde p})^p (1+\|z_t^m\|^k)^p\right]^{1/p}+ E\left[  (1+\|q_s^m\|^{\tilde p})^p(1+\|z_s^m\|^k)^p\right]^{1/p}\right.\notag\\
&\left.\hspace{13mm}+ \int_s^t E\left[(1+\|q_r^m\|^{\tilde p})^p(1+\|z_r^m\|^k)^p\right]^{1/p}dr \right).\notag
\end{align}
From this we can use Theorem \ref{theorem:conv} to find
\begin{align}
&\sup_{0\leq s\leq t\leq T}E\left[\left|J^m_{s,t}-\int_s^t B^{i_1,...,i_k}(r,q^m_r)  \left(\int h(r,q^m_r,z) z_{i_1}...z_{i_k} dz \right) dr\right|^p\right]^{1/p}\\
\leq & \tilde C( m^{1/2}T+ m(2+T)+m^{1/2} T^{1/2}) \sup_{r\in[0,T]} E\left[ (1+\|q_r^m\|^{\tilde p})^{2p}\right]^{1/(2p)}E\left[(1+\|z_r^m\|^{k+1})^{2p}\right]^{1/(2p)}\notag\\
=&O(m^{1/2}).\notag  
\end{align}

We can now compute   
\begin{align}
&\sup_{0\leq s\leq t\leq T}E\left[\left|J_{s,t}^m-J_{s,t}\right|^p\right]^{1/p}\\
\leq & O(m^{1/2})+\sup_{0\leq s\leq t\leq T}E\left[\left|\int_s^t B^{i_1,...,i_k}(r,q^m_r)  \left(\int h(r,q^m_r,z) z_{i_1}...z_{i_k} dz \right) dr\right.\right.\notag\\
&\left.\left.- \int_s^t B^{i_1,...,i_k}(r,q_r)  \left(\int h(r,q_r,z) z_{i_1}...z_{i_k} dz \right) dr\right|^p\right]^{1/p}\notag\\
\leq &O(m^{1/2})+C_{i_1...i_k} \int_0^T E\left[\left|(B^{i_1,...,i_k}\beta^{-k/2})(r,q_r^m)  - (B^{i_1,...,i_k}\beta^{-k/2})(r,q_r) \right|^p\right]^{1/p}dr\notag
\end{align}
where
\begin{align}
C_{i_1...i_k}=\left(\frac{1}{2\pi}\right)^{n/2}\int  e^{-\|w\|^2/2}    w_{i_1}...w_{i_k} dw.
\end{align}
The assumptions imply $B^{i_1,...,i_k}\beta^{-k/2}$ are $C^1$ with polynomially bounded first derivatives, and therefore the fundamental theorem of calculus can be used to show that
\begin{align}
\sup_{t\in[0,T]}|(B^{i_1,...,i_k}\beta^{-k/2})(t,q)-B^{i_1,...,i_k}\beta^{-k/2}(t,\tilde q))|\leq \tilde C(1+\|q\|^{\tilde p}+\|q-\tilde q\|^{\tilde p})\|q-\tilde q\|
\end{align}
for some $\tilde C,\tilde p>0$.  

Therefore, again  using Theorem \ref{theorem:conv}, we find
\begin{align}
&\sup_{0\leq s\leq t\leq T}E\left[\left|J_{s,t}^m-J_{s,t}\right|^p\right]^{1/p}\\
\leq &O(m^{1/2})+\tilde C T\sup_{r\in [0,T]} E\left[ (1+\|q_r\|^{\tilde p}+\|q_r-\tilde q_r^m\|^{\tilde p})^{2p}\right]^{1/(2p)}E\left[\|q_r^m-q_r\|^{2p}\right]^{1/{2p}}\notag\\
=&O(m^{1/2}).\notag
\end{align}
The result for general $p>0$ then follows from H\"older's inequality.

\end{proof}
\begin{corollary}
If the tensor rank, $k$, is odd then $J_{s,t}=0$ and hence
\begin{align}
\sup_{0\leq s\leq t\leq T}E\left[\left|J_{s,t}^m\right|^p\right]^{1/p}=O(m^{1/2})
\end{align}
as $m\rightarrow 0$.
\end{corollary}
 Processes of the form  $m^{-1/2}J^m_{s,t}$ for  $k$ is odd do  appear in the expression for the entropy production, \req{S_env_final}. The above corollary proves that they don't explode in the $L^p$ norm as $m\rightarrow 0$.  In fact, we will now prove that their expected values have a well behaved limit.

\begin{theorem}\label{homog_thm_2}
Let Assumption \ref{assump:beta_deriv} hold, $T>0$, $k$ be odd,  and $B:\mathbb{R}\times\mathbb{R}^n\rightarrow T^k(\mathbb{R}^n)$ be $C^2$  with $B$, $\partial_t B$, $\partial_{q^i} B$, $\partial_t\partial_{q^i}B$, and $\partial_{q^j}\partial_{q^i}B$ polynomially bounded in $q$, uniformly in $t\in [0,T]\times\mathbb{R}^n$ and consider the family of processes
\begin{align}
J_{s,t}^m=\int_s^t B^{i_1,...,i_k}(r,q_r^m) (z_r^m)_{i_1}...(z_r^m)_{i_k} dr
\end{align}
for $0\leq s\leq t\leq T$. Then, as  $m\rightarrow 0$, we have
\begin{align}
&\frac{1}{\sqrt{m}}E\left[J_{s,t}^m\right]\\
=&-\int_s^t  E\left[(-\nabla_q V(r,q_r)-\partial_r\psi(r,q_r)+\tilde F(r,q_r))\cdot \left(\int  (\nabla_z\chi)(r,q_r,z) h(r,q_r,z)dz\right)\right]dr\notag\\
&-\int_s^tE\left[  \int \left((\nabla_q\chi)(r,q_r,z )\cdot z\right) h(r,q_r,z)dz \right]dr+O(m^{1/2}),\notag
\end{align}
where   $h$ is given by \req{h_def2} and $\chi$ is defined from $B$ as in \req{chi_def}.
\end{theorem}
\begin{proof}

The hypotheses of Theorem \ref{homog_thm_1} hold, so we can follow its proof up to \req{Delta_J_formula} to obtain

\begin{align}\label{homog_eq_odd}
&m^{-1/2}J^m_{s,t}\\
=&-\int_s^t(\nabla_q\chi)(r,q_r^m,z_r^m)\cdot z_r^m dr\notag\\
&-\int_s^t (\nabla_z\chi)(r,q_r^m,z_r^m) \cdot \left[(-\nabla_q V(r,q_r^m)-\partial_r\psi(r,q^m_r)+\tilde F(r,q_r^m))dr+\sigma(r,q^m_r) dW_r\right]\notag\\
&+m^{1/2}\left(\chi(t,q_t^m,z_t^m)-\chi(s,q_s^m,z_s^m)-\int_s^t\partial_r \chi(r,q_r^m,z_r^m)dr\right),\notag
\end{align}
where $\chi$ is defined in \req{chi_def}.

The following computation shows that
\begin{equation}
M_{s,t}\equiv \int_s^t (\nabla_z\chi)(r,q_r^m,z_r^m) \cdot \sigma(r,q^m_r) dW_r
\end{equation}
 is a martingale (see \cite{karatzas2014brownian}):
\begin{align}\label{martingale_proof}
&E\left[\int_s^t\| (\nabla_z\chi)(r,q_r^m,z_r^m) \cdot \sigma(r,q^m_r)\|^2 dr\right]\\
\leq&\tilde C\|\sigma\|^2_{\infty}  E\left[ \int_s^t  (1+\|q_r^m\|^{\tilde p})^2(1+\|z_r^m\|^{k})^2dr\right]\notag\\
\leq&\tilde C\|\sigma\|^2_{\infty} (t-s)\sup_{r\in[0,t]} E\left[(1+\|q_r^m\|^{\tilde p})^4\right]^{1/2} E\left[(1+\|z_r^m\|^{k-1})^4\right]^{1/2}<\infty,\notag
\end{align}
where we used \req{chi_bounds}, Assumption \ref{assump:beta_deriv}, and Theorem \ref{theorem:conv}.

Therefore
\begin{align}
&m^{-1/2}E[J^m_{s,t}]\\
=&-E\left[\int_s^t (\nabla_z\chi)(r,q_r^m,z_r^m) \cdot (-\nabla_q V(r,q_r^m)-\partial_r\psi(r,q_r^m)+\tilde F(r,q_r^m))dr\right]\notag\\
&-E\left[\int_s^t(\nabla_q\chi)(r,q_r^m,z_r^m)\cdot z_r^m dr\right]+O(m^{1/2}),\notag
\end{align}
where we used the same reasoning as in the proof of \req{Delta_J_bounds1} to bound the last term.

$\nabla_z\chi(t,q,z) \cdot (-\nabla_q V(t,q)-\partial_t\psi(t,q)+\tilde F(t,q))$ and $z\cdot\nabla_q\chi(t,q,z)$ are both finite sums of multi-linear functions of $z$.  Tracing the definition \req{chi_def}, one can see that each tensor in the sum is a $C^1$ function of $(t,q)$ and has zeroth and first derivatives that are polynomially bounded in $q$, uniformly in $t\in[0,T]$.  Therefore Theorem \ref{homog_thm_1} applies to these integrals, giving
\begin{align}
&E\left[\int_s^t (\nabla_z\chi)(r,q_r^m,z_r^m) \cdot (-\nabla_q V(r,q_r^m)-\partial_r\psi(r,q_r^m)+\tilde F(r,q_r^m))dr\right]\\
=& E\left[\int_s^t  (-\nabla_q V(r,q_r)-\partial_r\psi(r,q_r)+\tilde F(r,q_r))\cdot \left(\int (\nabla_z\chi)(r,q_r,z)  h(r,q_r,z)dz\right)dr\right]+O(m^{1/2})\notag
\end{align}
and
\begin{align}
&E\left[\int_s^t(\nabla_q\chi)(r,q_r^m,z_r^m)\cdot z_r^m dr\right]\\
=&E\left[\int_s^t \int \left((\nabla_q\chi)(r,q_r,z)\cdot z\right) h(r,q_r,z)dzdr\right]+O(m^{1/2}).\notag
\end{align}
This completes the proof.

\end{proof}

\section{The Cell Problem}\label{app:lemma}

\setcounter{lemma}{0}
    \renewcommand{\thelemma}{\Alph{section}\arabic{lemma}}
This appendix details the solution to the cell problem,  \req{cell_prob}, a certain inhomogeneous linear partial differential equation that is useful for homogenizing integral processes. Specifically, we provide an explicit solution for the case where the inhomogeneity is a multi-linear function.

We will need the following lemma bounding the spectrum of a matrix.  See, for example, Appendix A in \cite{BirrellHomogenization} for a proof.
\begin{lemma}\label{eig_bound_lemma1}
Let $A$ be an $n\times n$ real or complex matrix with symmetric part $A^s=\frac{1}{2}(A+A^*)$. If the eigenvalues of $A^s$ are bounded above (resp. below) by $\alpha$ then the real parts of the eigenvalues of $A$ are bounded above (resp. below) by $\alpha$.
\end{lemma}

We will also need the following result, which solves a kind of generalized Lyapunov equation.
\begin{lemma}\label{gen_liap_lemma}
Let $V$ be a finite dimensional vector space over $\mathbb{C}$, $C:V\to V$ be linear, and $B:V^k\rightarrow\mathbb{C}$ be multi-linear (i.e. $B\in T^k(V)$). If the eigenvalues of $C$ all have negative real parts then there exists a unique $A\in T^k(V)$ that satisfies $\sum_i A(\cdot,...,\cdot,C\cdot,\cdot,...,\cdot)=-B$ (i.e. for the $i$th term in the sum, the $i$th input is composed with $C$). $A$ is given by
\begin{align}\label{A_def}
A(v_1,...,v_k)=\int_0^\infty B(e^{tC}v_1,...,e^{tC}v_k) dt.
\end{align}

\end{lemma}
\begin{proof}
The eigenvalue bound implies the existence of  $\tilde C>0$, $\mu>0$ such that $\|e^{tC}\|\leq \tilde Ce^{-\mu t}$, therefore the integral \req{A_def} exists. We have
\begin{align}
&\sum_i A(v_1,...,Cv_i,...,v_k)=\sum_i \int_0^\infty B(e^{tC}v_1,...,e^{tC}Cv_i,...,e^{tC}v_k) dt\\
=&\int_0^\infty \frac{d}{dt}B(e^{tC}v_1,...,e^{tC}v_i,...,e^{tC}v_k) dt\notag\\
=&-B(v_1,...,v_k).\notag
\end{align}
Therefore \req{A_def} provides the desired solution.

To prove uniqueness, it suffices to show that $A=0$ is the unique solution corresponding to $B=0$. To this end, suppose  $\sum_i A(\cdot,...,\cdot,C\cdot,\cdot,...,\cdot)=0$.  Let $\lambda_i$ be eigenvalues of $C$ and $e^i_j$ be a basis of generalized eigenvectors, where $\{e^i_j\}_j$ is a basis for the eigenspace corresponding to $\lambda_i$ and $C e^i_j=\lambda_ie^i_j+e^i_{j-1}$ ($e_{-1}\equiv 0)$. Then
\begin{align}
0=\sum_l A(e^{i_1}_{0},...,Ce^{i_l}_{0},...,e^{i_k}_0)=\left(\sum_l \lambda_{i_l} \right)A(e^{i_1}_0,...,e^{i_k}_0).
\end{align}
The coefficient is non-zero since the real parts of the $\lambda_i$ are all negative.  Therefore 
\begin{align}
A(e^{i_1}_0,...,e^{i_k}_0)=0.
\end{align}

We now show   $A(e^{i_1}_{j_1},...,e^{i_k}_{j_k})=0$ for all choices of $i$'s and $j$'s. This will prove that $A=0$ by multi-linearity and the fact that the $e^{i}_j$'s form a basis.  We induct on $N=\sum_{l} j_l$. We showed it above for $N=0$.  Suppose it holds for $N-1$.  Given $j_l$ with $\sum_l j_l=N$ we have
\begin{align}
0=&\sum_l A(e^{i_1}_{j_1},...,Ce^{i_l}_{j_l},...,e^{i_k}_{j_k})\\
=&\left(\sum_l \lambda_{i_l} \right)A(e^{i_1}_{j_1},...,e^{i_k}_{j_k})+\sum_lA(e^{i_1}_{j_1},...,e^{i_l}_{j_l-1},...,e^{i_k}_{j_k}).\notag
\end{align}
$j_1+...+(j_l-1)+...+j_k=N-1$, so the last term vanishes by the induction hypothesis.  As before, $\sum_l \lambda_{i_l}\neq 0$, hence $A(e^{i_1}_{j_1},...,e^{i_k}_{j_k})=0$.  This proves the claim by induction.
\end{proof}

Finally, the following lemma details the solution to  the cell problem, \req{cell_prob}.
\begin{lemma}\label{lemma:cell_prob}
Consider the differential operator $L$ defined by
\begin{align}\label{L_app_def2}
(L\chi)(z)=&\beta^{-1} \gamma_{\xi\zeta}(\partial_{z_\xi}\partial_{z_\zeta}\chi)(z)-\tilde\gamma_{\xi\eta}\delta^{\eta\zeta}z_\zeta(\partial_{z_\xi}\chi)(z)
\end{align}
where  $\gamma$, the symmetric part of $\tilde \gamma$, is positive definite and $\beta>0$.

Let $k\geq 1$ and $B\in T^k(\mathbb{R}^n)$.  For $j=0,...,\lfloor (k-1)/2\rfloor$ define $A_j \in T^{k-2j}(\mathbb{R}^n)$ inductively by
\begin{align}\label{A0_def}
A_0(v_1,...,v_k)=-\int_0^\infty B(e^{-t\tilde\gamma}v_1,...,e^{-t\tilde\gamma}v_k)dt
\end{align}
and
\begin{align}\label{Aj_def}
A_{j}(v_1,...,v_{k-2j})=\int_0^\infty 2\beta^{-1} \sum_{\alpha=1}^{k-2(j-1)-1}  \sum_{\delta>\alpha}  A_{j-1}^{\alpha\delta} (e^{-t\tilde\gamma}v_1,...,e^{-t\tilde\gamma}v_{k-2j})dt,
\end{align}
where $A_j^{\alpha\delta}\in T^{{k-2(j+1)}}(\mathbb{R}^n)$ is the multi-linear map with components $ A_j^{i_1,...,i_{k-2j}}   \gamma_{i_\alpha i_\delta}$ and, for the purposes of taking the operator exponential, $\tilde\gamma$ is to be thought of as the linear map with action $z_\xi\to \tilde\gamma_{\xi\eta}\delta^{\eta\zeta}z_\zeta$.

Then
\begin{align}
\chi(z)=\sum_{j=0}^{\lfloor (k-1)/2\rfloor} A_j(z,...,z)
\end{align}
is a solution to the cell problem
\begin{align}\label{eq:cell_prob}
(L\chi)(z)=B(z,...,z)-\left(\frac{\beta}{2\pi}\right)^{n/2}\int B(\tilde z,...,\tilde z)e^{-\beta\|\tilde z\|^2/2}d\tilde z.
\end{align}
Note that, if $k$ is odd, the integral in \req{eq:cell_prob} vanishes.

Consider the components $A_j^{i_1,...,i_{k-2j}}$ to be functions of $(\beta,\tilde\gamma,B)$, defined on the domain where  $\beta>0$, $\tilde\gamma$ has positive definite symmetric part, and $B\in T^k(\mathbb{R}^n)$.  The $A_j^{i_1,...,i_{k-2j}}$  are $C^\infty$  jointly in all of their variables on this domain and are linear in $B$. 

Let $U_{\epsilon, R}$ be the open set defined by $\beta>\epsilon$ and the symmetric part of $\tilde \gamma$ having eigenvalues in the interval $(\epsilon,R)$.  Given $B\in T^k(\mathbb{R}^n)$, any order derivative (including the zeroth) of  $(\beta,\tilde\gamma)\to A_j^{i_1,...,i_{k-2j}}(\beta,\tilde\gamma,B)$ with respect to any combination of its variables is bounded by $\tilde C \|B\|$  on $U_{\epsilon,R}$ for some $\tilde C>0$ ($\tilde C$ depends on $\epsilon$, $R$, and the choice of derivatives, but not on $B$).

\end{lemma}
\begin{proof}
 For $j=0,...,\lfloor (k-1)/2\rfloor$ let $A_j \in T^{k-2j}(\mathbb{R}^n)$ be defined by \req{A0_def}-\req{Aj_def}. Note that Lemma \ref{eig_bound_lemma1} implies that the real parts of the eigenvalues of $-\tilde\gamma$ are negative, and hence the integrals in the definitions exist.
 
 Define
\begin{align}
\chi(z)=\sum_{j=0}^{\lfloor (k-1)/2\rfloor} A_j(z,...,z).
\end{align}
We have
\begin{align}
(\partial_{z_\xi} A_j)(z,...,z)=& \sum_{\alpha=1}^{k-2j} A_j^{i_1,...,i_{k-2j}}z_{i_1}...z_{i_{\alpha-1}} \delta_{i_\alpha }^\xi z_{i_{\alpha+1}}...z_{i_{k-2j}},\\
(\partial_{z_\zeta}\partial_{z_\xi} A_j)(z,...,z)=&A_j^{i_1,...,i_{k-2j}} \sum_{\alpha=1}^{k-2j-1}  \sum_{\delta>\alpha} \left( \delta_{i_\alpha}^{ \xi} \delta_{i_\delta}^\zeta + \delta_{i_\delta}^\xi \delta_{i_\alpha}^\zeta\right) \prod_{\rho\neq \alpha,\delta}z_{i_{\rho}},\notag
\end{align}
and so
\begin{align}
(L\chi)(z)=&\sum_{j=0}^{\lfloor (k-1)/2\rfloor} \left(\beta^{-1} \gamma_{\xi\zeta}(\partial_{z_\xi}\partial_{z_\zeta}A_j)(z)-\tilde\gamma_{\xi \eta}\delta^{\eta\zeta}z_\zeta(\partial_{z_\xi}A_j)(z)\right)\\
=&\sum_{j=0}^{\lfloor (k-1)/2\rfloor} \left(2\beta^{-1} \sum_{\alpha=1}^{k-2j-1}  \sum_{\delta>\alpha}  A_j^{\alpha\delta} (z,...,z)-\sum_{\alpha=1}^{k-2j}  A_j(z,...,z,\tilde \gamma z,z,...,z)\right),  \notag
\end{align}
where $A_j^{\alpha\delta}\in T^{{k-2(j+1)}}(\mathbb{R}^n)$ is the multi-linear map with components $ A_j^{i_1,...,i_j}   \gamma_{i_\alpha i_\delta}$ and it is the  $\alpha$'th input of $A_j(z,...,z,\tilde \gamma z,z,...,z)$ that equals $\tilde \gamma z$ in the above sum.

Collecting terms involving tensors of the same degree, we have
\begin{align}
&(L \chi)(z)-B(z,...,z)\\
=&\sum_{j=1}^{\lfloor (k-1)/2\rfloor}\left( 2\beta^{-1} \sum_{\alpha=1}^{k-2(j-1)-1}  \sum_{\delta>\alpha}  A_{j-1}^{\alpha\delta} (z,...,z)-\sum_{\alpha=1}^{k-2j}  A_j(z,...,z,\tilde \gamma z,z,...,z)\right)  \notag\\
&- \left(\sum_{\alpha=1}^{k}  A_0(z,...,z,\tilde \gamma z,z,...,z)+B(z,...,z)\right)\notag\\
&+2\beta^{-1} \sum_{\alpha=1}^{k-2\lfloor (k-1)/2\rfloor-1}  \sum_{\delta>\alpha}  A_{\lfloor (k-1)/2\rfloor}^{\alpha\delta} (z,...,z).\notag
\end{align} 

Recalling the definition of $A_0$ and $A_j$ from \req{A0_def} and \req{Aj_def}, Lemma \ref{gen_liap_lemma}  implies that they satisfy 
\begin{align}
&\sum_{\alpha=1}^{k}  A_0(z,...,z,\tilde\gamma z,z,...,z)=-B(z,...,z),\\
&\sum_{\alpha=1}^{k-2j} A_j(z,...,z,\tilde \gamma z,z,...,z)= 2\beta^{-1} \sum_{\alpha=1}^{k-2(j-1)-1}  \sum_{\delta>\alpha}  A_{j-1}^{\alpha\delta}(z,...,z),\hspace{2mm} j\geq 1. 
\end{align}
Therefore
\begin{align}\label{L_chi_1}
&(L \chi)(z)-B(z,...,z)=2\beta^{-1} \sum_{\alpha=1}^{k-2\lfloor (k-1)/2\rfloor-1}  \sum_{\delta>\alpha}  A_{\lfloor (k-1)/2\rfloor}^{\alpha\delta} (z,...,z).
\end{align} 

If $k$ is odd then $k-2\lfloor (k-1)/2\rfloor=1$ and therefore the second summation in \req{L_chi_1} is empty.  This gives $L\chi(z)=B(z,...,z)$ as claimed. If $k$ is even then $k-2\lfloor (k-1)/2\rfloor=2$ and $A_{\lfloor (k-1)/2\rfloor}^{\alpha\delta}\in T^0(\mathbb{R}^k)=\mathbb{R}$. Therefore the right hand side of \req{L_chi_1}, call it $\tilde B$, is a constant.

The value of $\tilde B$ can be computed by integrating both sides against
\begin{align}
h(z)=\left(\frac{\beta}{2\pi}\right)^{n/2}e^{-\beta\|z\|^2/2}.
\end{align}
Using the fact that $\int h(z)dz=1$ results in
\begin{align}
\tilde B=\int (L\chi)(\tilde z) h(\tilde z) d\tilde z-\int B(\tilde z,...,\tilde z) h(\tilde z)d\tilde z.
\end{align}
Integrating by parts, observing that the boundary terms vanish at infinity, and using  $L^*h=0$, where $L^*$ is the formal adjoint of $L$ we find
\begin{align}
\tilde B=-\int B(\tilde z,...,\tilde z) h(\tilde z)d\tilde z
\end{align}
as claimed.

We now prove the claimed smoothness and boundedness properties. Let $U$ be the subset of the $n\times n$ real matrices such that all of the eigenvalues of the symmetric part of the matrix are negative.  This is an open set and the functions $G^{i_1...i_l}_{j_1....j_l}:U\rightarrow\mathbb{\mathbb{R}}$,
\begin{align}
G^{i_1...i_l}_{j_1....j_l}(A)=\int_0^\infty (e^{tA})^{i_1}_{j_1}...(e^{tA})^{i_l}_{j_l}dt,
\end{align}
are smooth and can be differentiated under the integral.  Restricted to the subset where the eigenvalues of the symmetric part are less than $-\epsilon<0$, $G^{i_1...i_l}_{j_1....j_l}$ and its derivatives are all bounded.  These facts can be proven by using the dominated convergence theorem, along with the formula for the derivative of the matrix exponential found in \cite{exp_deriv}.

Therefore 
\begin{align}
A_0^{i_1,...,i_k}(\beta,\tilde\gamma,B)=-B^{j_1,...,j_k}G^{i_1...i_k}_{j_1....j_k}(-\tilde\gamma)
\end{align}
which is linear in $B$, smooth in $(B,\tilde\gamma)$, and it, along with its derivatives, are bounded by $\tilde C\|B\|$,  on the domain where the symmetric part of $\tilde\gamma$ has eigenvalues contained in $(\epsilon,R)$.

Now, assume $A_{j-1}$ satisfies the desired properties.  Then one can easily verify that
\begin{align}
&A_{j}^{i_1,...,i_{k-2j}}(\beta,\tilde\gamma,B)\\
=&2\beta^{-1} \sum_{\alpha=1}^{k-2(j-1)-1}  \sum_{\delta>\alpha} A_{j-1}^{l_1...l_{\alpha-1} \eta l_{\alpha}.... l_{\delta-2} \xi l_{\delta-1}...l_{k-2j}}(\beta,\tilde\gamma,B)\frac{1}{2} (\tilde\gamma_{\eta\xi}+\tilde\gamma_{\xi\eta})G_{i_1...i_{k-2j}}^{l_1...l_{k-2j}}(-\tilde\gamma)\notag
\end{align}
does as well. Therefore the claim holds for all $j$ by induction.

\end{proof}

\subsection*{Acknowledgments}

Many thanks to J. Wehr for bringing this problem to my attention and for numerous stimulating discussions.  

\bibliographystyle{unsrt}
\bibliography{refs}

\end{document}